\documentclass[11pt]{article}
\usepackage{amsmath,amssymb,amsfonts,amsthm,amscd,color}
\usepackage{graphicx}
\usepackage{hyperref}
\newcommand{\defn}[1]{\emph{#1}}

\newcommand{\me}{e}
\newcommand{\mi}{i}

\renewcommand{\Re}{\operatorname{Re}}
\renewcommand{\Im}{\operatorname{Im}}
\newcommand{\sign}{\operatorname{sign}}

\newtheorem{theorem}{Theorem}[section]
\newtheorem{lemma}[theorem]{Lemma}

\newtheorem{proposition}[theorem]{Proposition}
\newtheorem{corollary}[theorem]{Corollary}

\theoremstyle{definition}
\newtheorem{definition}[theorem]{Definition}
\theoremstyle{remark}
\newtheorem{remark}[theorem]{Remark}

\begin{document}

\title{Reconstruction phases in the planar three- and four-vortex problems}
\author{Antonio Hern\'andez-Gardu\~no \\
\normalsize Independent scholar, Mexico \\
{\footnotesize \texttt{antonio.hega@gmail.com}}\\
\\
Banavara N. Shashikanth \\
\normalsize Department of Mechanical and Aerospace Engineering \\
\normalsize New Mexico State University \\
\normalsize MSC 3450, P.O. Box 30001 \\
\normalsize 88003-8001 Las Cruces NM, U.S.A. \\
{\footnotesize \texttt{shashi@nmsu.edu}}\\
}

\date{September 15, 2017}
\maketitle

\begin{abstract}
Pure reconstruction phases---geometric and dynamic---are computed in the $N$-point-vortex model in the plane, for the cases $N=3$ and $N=4$. The phases are computed relative to a metric-orthogonal connection on appropriately defined  principal fiber bundles. The metric is similar to the kinetic energy metric for point masses but with the masses replaced by vortex strengths. The geometric phases are shown to be proportional to areas enclosed by the closed orbit on the symmetry reduced spaces. More interestingly, simple formulae are obtained for the dynamic phases, analogous to Montgomery's result for the free rigid body, which show them to be proportional to the time period of the symmetry reduced closed orbits.
For the case $ N = 3 $ a non-zero total vortex strength is assumed.  For the case $ N = 4 $ the vortex strengths are assumed equal.

\bigskip
\noindent
\small{{\textbf{Keywords}}: 3-vortex problem, 4-vortex problem, reconstruction, holonomy, geometric phase, dynamic phase.}
\end{abstract}

\tableofcontents

\section[Introduction]{Introduction:  symmetry and reconstruction \\
phases} 

Reconstruction phases are drifts in symmetry group directions at the end of a periodic evolution on the symmetry reduced phase space/configuration space of a dynamical system.\footnote{The occurrence of the word `phase' in  both `reconstruction phase' and `phase space' is, of course, merely a coincidence.} Mathematically, reconstruction phases are elements of the symmetry Lie group of a dynamical system. These elements typically consist of a {\it dynamic phase}, a {\it geometric phase} and an element that is their group composition termed the {\it total phase}. In the case of a vector group, one thus obtains the paradigm:
\[Total \; phase=Geometric \; phase \; + \; Dynamic \; phase\]
Much of the interest in phases has focused, naturally, on the geometric phase. The geometric phase is determined by the choice of a purely geometric construction on the principal fiber bundle structure of the system's (unreduced) phase space/configuration space, called {\it a connection}. The geometric phase is, in fact, an element of the holonomy subgroup determined by the connection and is therefore also known
 as the holonomy of the connection. It depends on geometric features of the periodic orbit such as, for example, the area enclosed by it, and is independent of dynamic features such as time parametrization of the orbit.

Since the total phase, for a given evolution, is independent of the connection, the choice of the connection is thus equivalent to a choice of the above splitting of the total phase into a geometric and a dynamic part. Typically, however, in many systems of interest there exists a naturally occurring connection. For example, in simple mechanical systems the system kinetic energy equips configuration space with a metric which defines a connection via metric orthogonality.

Despite (a) the fact that the notion of holonomy in differential geometry has existed independently of applications to phases, see, for example, Kobayashi and Nomizu \cite{KoNo96}, (b) the well-known existence of at least one classical dynamical system exhibiting phases---the Foucault pendulum---and (c) the existence of Pancharatnam's phase in classical optics  \cite{Pa56, Be87} and the Aharanov-Bohm effect in quantum mechanics \cite{AhBo61}, the interest in phases, in particular geometric phases, was really ignited by Berry's paper \cite{Be84}.  Berry's geometric phase was in the context of parameter-dependent quantum Hamiltonian systems in which the parameters undergo an infinitely slow, or adiabatic, cyclic evolution.\footnote{For this reason perhaps there still exists a misconception that geometric phases require adiabatic evolutions for their existence.}  The relation of Berry's phase to holonomy was clarified by Simon \cite{Si83}. Around the same time Guichardet \cite{Gu1984} proved the existence of a geometric phase in a model of classical molecular dynamics by showing that vibratory motions can lead to net rotations, also pointing out the relation to the `falling cat' problem; see also Iwai \cite{Iw1987}. The most comprehensive treatment of phases and their relation to holonomy, connections and the general subject of dynamical systems with symmetry is the work of Marsden, Montgomery and Ratiu \cite{MaMoRa90}. Following their terminology, we will distinguish between {\it adiabatic phases} and {\it pure reconstruction phases}. The former require adiabatic evolutions and the latter do not.  The lack of the adiabaticity requirement makes phases significantly more interesting since it allows for the possibility of observation or experimental verification. Indeed the Foucault pendulum phase and Pancharatnam's phase are two such examples.

In most papers on phases, little or no attention is paid to obtaining simple exact formulas for the dynamic phase. This is not surprising since the dynamic phase depends on the system's dynamics and for a nonlinear dynamical system such formulas may not even exist. A noteworthy exception is Montgomery's work on phases in the dynamics of a free rigid body where, in addition to a geometric phase formula, a simple exact formula for the dynamic phase is obtained; see \cite{Mo91}. In the framework of Poisson reduction and in the framework of almost K\"{a}hler manifolds, respectively, Blaom \cite{Bl2000} and Pekarsky and Marsden \cite{PeMa2001} have also obtained some general formulas for dynamic and geometric phases. 

This paper focuses on pure reconstruction phases in classical fluid mechanics. In particular, we examine phases in a popular model of coherent vorticity in incompressible, homogeneous, inviscid flows---the $N$-point-vortex model. The cases $N=3$ and $N=4$ are considered.  Adiabatic phases have been computed in this model previously by Newton \cite{Ne94} and by Shashikanth and Newton \cite{ShNe98}, and also in a model of interacting vortex patches by Shashikanth and Newton \cite{ShNe2000}. Pure reconstruction phases have been computed in an axisymmetric model of vortex rings in $\mathbb{R}^3$ by Shashikanth and Marsden \cite{ShMa2003}. To the best of our knowledge, pure reconstruction phases in the $N$-point-vortex model, in particular dynamic phases, have not been computed before.

The main results are summarized below. 

\begin{enumerate}
\item In the three vortex problem, formulas for the geometric and dynamic phase (pure reconstruction phases) are obtained for vortices of arbitrary strengths, assuming that the total vortex strength and the angular impulse are different from zero.

\item In the four vortex problem, formulas for the geometric and dynamic phase (pure reconstruction phases) are obtained for vortices of equal strengths.

\item Simple exact formulas are obtained for the dynamic phase which show that in both cases it is proportional to the time period of the closed periodic orbit on the respective symmetry reduced spaces, analogous to Montgomery's formula for the free rigid body.
\end{enumerate}

The outline of the paper is as follows. In the next section, brief introductions to the $N$-point-vortex model and to the differential geometric and geometric mechanics concepts necessary for understanding reconstruction phases are provided. Sections \ref{sec:phasesThreeVortices} and \ref{sec:phasesFourVortices} are devoted to the computation of phases in the three and four vortex problems, respectively. Section \ref{sec:summary} presents a summary and future directions. The main body of the paper is followed by two appendices on different aspects of the three vortex problem, which are not essential for an understanding of the main results of the paper and may be skipped in a first reading.

\section{Preliminaries}

In this section, some preliminaries related to (i) the $N$-point-vortex model and (ii) the differential geometric  notions of principal fiber bundles, connections and holonomy, are established.

\subsection{The $N$-point-vortex model in the plane.}

The $N$-point-vortex model in the plane is a model for the dynamics of  coherent vorticity in incompressible, homogeneous and inviscid fluid flows governed by Euler's equation in $\mathbb{R}^2$. In this model, it is assumed that the vorticity is concentrated at $N$ points. The vorticity field $\omega$, using the complex representation of the plane, is expressed as the sum of $N$ Dirac delta functions centered on $z_\alpha=x_\alpha+ \mi y_\alpha$, the locations of the point vortices,
\[\omega(z)=\sum_{\alpha=1}^N \Gamma_\alpha \delta(z-z_\alpha),\]
where the $\Gamma_\alpha$'s are the strengths of the  vortices.

We now recount the basic description of the planar $N$-vortex problem as a Hamiltonian system.  For details the reader may consult the comprehensive treatise \cite{Ne2001}.  

\subsubsection*{Equations of motion, Hamiltonian formulation and symmetries.}

The equations of motion of $N$ point vortices are
	\begin{displaymath}
	\dot{ z} _\alpha = \frac{ \mi }{ 2 \pi} \sum _{ \beta \neq \alpha}^N \Gamma _\beta \frac{ z _\alpha - z _\beta }{ |z _\alpha - z _\beta | ^2 }
	\end{displaymath} These equations are derived from Euler's equation in $\mathbb{R}^2$ through its formulation in terms of the vorticity field
$$ \frac{ D \mathbf{\omega} }{ Dt } := \mathbf{ \omega } _t + \mathbf{u} \cdot \nabla \mathbf{\omega } = 0 \,. $$

The equations of motion are equivalent to Hamilton's equations $ \mathbf{ i} _{ X _h } \Omega = d h $ with Hamiltonian
  $$ h = - \frac{ 1 }{ 2 \pi } \sum _{ \alpha < \beta} \Gamma _\alpha \Gamma _\beta \ln | z _\alpha - z _\beta | \;, $$
and symplectic form
\begin{equation} \label{vortexSymplecticForm}
	\Omega _0 (z, w) = - \operatorname{ Im} \sum _{ \alpha = 1} ^n \Gamma _\alpha z _\alpha \bar{ w} _\alpha \;. 
\end{equation}

The Hamiltonian and symplectic form are invariant with respect to the diagonal action of $ SE(2) $ on the phase space of the system identified  with $\mathbb{C}^N$, or $\mathbb{R}^{2N}$, (minus collision points):
\[
  z _i \mapsto \me ^{ \mi \theta} z _i + a
\]
\[
  (\theta, a) \in SE(2) \cong S ^1 \times \mathbb{C}
\]
The $N$-vortex model admits three conserved quantities related to the $ SE(2) $ symmetries and two more, related to time-translation and rescaling symmetries:
\vspace{-1ex}
\begin{center}
\newcommand{\z}{\small}\newcommand{\bz}{\bf\small}
\begin{tabular}{|c|c|c|}
\hline
\bz symbol & \bz name & \bz symmetry \\
\hline\hline
  $ Z _0 $ & \z center of circulation (two quantities) &  \z translational \\
\hline
  $ \Theta _0 $ & \z second moment of circulation\footnotemark & \z rotational \\
\hline
  $ \Psi _0 $ & \z finite part of kinetic energy & \z temporal \\
\hline
  $ V _0 $ & \z virial & \z rescaling \\
\hline
\end{tabular}
\end{center}
\begin{align}
  Z _0&= \Gamma_{\text{tot}} ^{-1} \sum _k \Gamma _k z _k \,, \quad \Theta _0 = \sum _k \Gamma _k | z _k | ^2 \,, \label{eq:circonstants}
\end{align}
\[
  \Psi _0= - \sum _{ n<k} \Gamma _n \Gamma _k \ln | z _n - z _k | \,,
\]
\begin{align}
  V_0 &= \frac1{2i} \sum _k \Gamma _k (\bar{z} _k \dot{z} _k - z _k \dot{ \bar{z}} _k ) = \sum _{ n < k} \Gamma _n \Gamma _k \,. \label{eq:virial}
\end{align}

\footnotetext{In the point vortex literature, this is also referred to as the \emph{angular impulse}.}

The expression $$ M = \sum _{ n < k} \Gamma _n \Gamma _k | z _n - z _k | ^2 $$ is also a conserved quantity.  It is not independent, since it is expressed in terms of $ \Theta_0 $ and $ Z _0 $:
\[
  M = \Gamma_{\text{tot}} \Theta _0 - \Gamma_{\text{tot}} ^2 | Z _0 | ^2 \,,
\]
where \[\Gamma_{\text{tot}}:=\sum_{k} \Gamma_k \,.\]

\subsection{Principal fiber bundles, connections and holonomy.}

In this section, basic differential geometric notions of principal fiber bundles, connections and holonomy are recalled. For a more detailed treatment of these notions, the reader is referred to the classic books by Kobayashi and Nomizu \cite{KoNo96}.  The discussion on connections and holonomy, and their relation to phases, is a synopsis of material from the texts \cite{Ma92, MaMoRa90}; see also  \cite{ChJa2004}.

Consider $G$ a compact Lie group acting freely on a manifold $E$,
\[
  \varphi : G \times E \longrightarrow E \,.
\]
The standard shorthand $ g \cdot p := \varphi _g (p) := \varphi(g,p) $ will be used.
Let $ B = E/G $, the space of group orbits.  Then $B$ is a smooth manifold and the canonical projection
	\[
	  \pi : E \longrightarrow B
	\]
is a smooth submersion.  We call $E$ a principal fiber bundle over $B$.  
The fibers of $\pi$ are diffeomorphic to $G$.  Moreover, the space $ V:= \ker {\bf D} \pi $ forms a vector subbundle (of the tangent bundle $TE$)  called the vertical subbundle. Elements of $V$ are tangent to the fibers and are called vertical vector fields.

Each $ \xi \in \mathfrak{g} $ (the Lie algebra of $G$), induces a vector field on $E$, the infinitesimal generator of the group action, defined by
	\[
	  \xi _E (p) = \left. \frac{d}{dt} \right| _{ t = 0 } \exp(t \xi) \cdot p \,, \quad p \in E \,.
	\]
Every element of the vertical subbundle $V$ is of this form.

An \defn{Ehresmann connection} on $E$ is a smooth subbundle $H$ of $ TE $ such that
	\[
	  T E = H \oplus V \,.
	\]
If the condition 
\begin{align}
	  H _{ (g \cdot p ) } &= d \varphi _g \cdot H _p      \label{eq:equiv} 
\end{align}
is satisfied then $H$ is a \defn{principal connection}.  Associated to an Ehresmann connection we have the vertical and horizontal operators:
	\[\begin{split}
	  \operatorname{ Ver} &: T _p E \longrightarrow V _p \\
	  \operatorname{ Hor} &: T _p E \longrightarrow H _p
	\end{split}\]
defined by the splitting
	\[
	  v _p = \operatorname{ Ver} (v _p) + \operatorname{ Hor} (v _p) \,.
	\]

The $ \operatorname{ Ver} $ operator induces a $\mathfrak{g}$-valued one-form $\alpha$ on $E$ defined by
	\begin{equation} \label{ConnectionProperties}
	  \alpha (\operatorname{ Ver}(v _p)) = \xi, \quad  \alpha (\operatorname{ Hor}(v _p)) = 0 \,, 
	\end{equation}
where $\xi _E (p) = \operatorname{ Ver} (v _p )$.  If $H$ is a principal connection, the connection one-form also satisfies
\begin{align}
	(\varphi_{g}^* \alpha)(v_p)=\operatorname{Ad}_g (\alpha(v_p)) \,, \label{eq:equiv1form} 
\end{align}
which is the equivalent of the equivariance condition~(\ref{eq:equiv}).

\subsubsection*{Holonomy}

We say that a curve $c_b$ on $B$ is lifted to a curve $ \bar{c} $ on $E$ if $ \bar{c} \, {} ' (t) $ is horizontal and $ \pi (\bar{c}(t)) = c_b(t) $  for all $t$.   An element on the fiber over $ c_b(t) $ is \defn{parallel transported} along horizontal lifts.

Given a closed curve (a \emph{loop})\footnote{Assume that time has been rescaled so that a period occurs over the time interval $ [0,1] $.} on the base, $ c_b: [0,1] \longrightarrow B $, the \defn{holonomy} of the connection at $ p = c_b(0) = c_b(1) $ is the operator $ \operatorname{ Hol} : F _p \longrightarrow F _p $ induced by parallel transport, where $ F _p := \pi ^{-1} (p) $.  When $G$ is abelian, the holonomy of a loop on $B$ is (cf. \cite{MaMoRa90}):
\begin{align}\label{eq:gphasegen}
	  g(1) &:= \exp \left[ - \int _0^1 (\sigma ^\ast \alpha )(c_b'(s))ds \right]= \exp \left[ - \iint \sigma ^\ast d \alpha \right] \,, 
\end{align}
where $\sigma: U \subset B \rightarrow E$ is a local section ($c_b \subset U$). The second equality, which follows from Stokes' theorem, holds only if the loop $c$ is the boundary of a two-dimensional submanifold of $B$. For Abelian groups, $\operatorname{Ad}_g=\operatorname{Id}$, and the equivariance condition for the connection 1-form ensures that the definition of $g(1)$ is independent of the choice of section.

\subsubsection*{Reconstruction phases}

A given (arbitrary) connection on a principal fiber bundle allows one to solve the \emph{dynamic reconstruction problem} (i.e. the opposite of reduction) for a given curve in $B$.  For Hamiltonian systems with symmetry the analysis proceeds as follows.  Let $ E = \mathcal{P} $, where $ (\mathcal{P}, \omega) $ is a symplectic manifold on which the Lie group $G$ acts on the left by canonical transformations in a Hamiltonian fashion, so that it induces an equivariant momentum map $ \mathbf{J} : \mathcal{P} \longrightarrow \mathfrak{g}^{\ast} $.  (See \cite{MaRa99} for background on momentum maps and related concepts.)  Further, suppose that $ \mu \in \mathfrak{g}^{\ast} $ is a regular value of $\mathbf{J}$, so that $ \mathbf{J} ^{-1}(\mu) $ is a smooth manifold, and that $ G _\mu $ (the isotropy subgroup of $\mu$) acts freely on $ \mathbf{J} ^{-1} (\mu) $.  Then $ \mathbf{J} ^{-1} (\mu) \longrightarrow \mathcal{P} _\mu $ is a $ G _\mu $-bundle with base space
\begin{equation} \label{JmuPmuGmu}
	\mathcal{P} _\mu := \mathbf{J} ^{-1} (\mu) / G _\mu \,,
\end{equation} 
which is called the \defn{reduced phase space}.

\begin{remark}\label{not_equivariant_momentum_map}
	The symmetry group for the $N$-vortex problem is $ SE(2) $ and it does not admit an equivariant momentum map unless $ \Gamma_{\text{tot}}=0 $ (see \cite{AdRa1988}).  However, for the case $\Gamma_{\text{tot}} \neq 0$, a \emph{reduction by stages} can be performed where one first fixes the center of circulation at the origin (this center is guaranteed to exist when $\Gamma_{\text{tot}} \neq 0$).  The remaining symmetry is encoded by $ SO(2) $, and its associated momentum map is indeed equivariant.
\end{remark}

Let $ H \in C ^{\infty} (\mathcal{P}) $ (the \emph{Hamiltonian function}) be $G$-invariant.  Then, by the \emph{Marsden-Weinstein reduction theorem} \cite{MaWei74}, the Hamiltonian system $ (\mathcal{P}, \omega, H) $ induces a \emph{reduced} Hamiltonian system $ (\mathcal{P} _\mu, \omega _\mu, H _\mu) $, with $ \mathcal{P}_\mu $ given by \eqref{JmuPmuGmu}.  (See \cite{Ma92} for a concise exposition of symplectic reduction.)

Now, suppose that $ c _\mu : I \longrightarrow \mathcal{P}_\mu $, where $ I \subset \mathbb{R}^+$, is a solution trajectory starting at $c_{\mu 0}$ of the Hamiltonian vector field associated to $ H _\mu $; i.e. it solves the initial value problem
\[
	\frac{d}{dt} c_\mu (t) = X _{ H_\mu } (c _\mu) 
\]
with $c_{\mu}(0)=c_{\mu 0}$.
 One then wishes to find the solution trajectory $ c:I \longrightarrow \mathcal{P}$ starting at $c_0$ of the original Hamiltonian vector field $H$, i.e. solving the initial value problem
\begin{equation} \label{originalEqMot}
	\frac{d}{dt} c(t) = X _H (c) \,, 
\end{equation}
with $c(0)=c_0$, and such that $\pi_\mu(c_\mu(t))=c(t), \forall t \in I$. This is the \emph{reconstruction problem}.

Given an arbitrary principal connection $\alpha$ on $ \pi_\mu $, the reconstruction problem can be solved in two steps.  First, a horizontal lift $ \bar{c} (t) $ of $ c _\mu (t) $ is found, with $ \bar{c} (0)=c_0$.  Second, the $\mathfrak{g}$-valued function $ \xi(t):= \alpha (X _H (\bar{c} (t))) $ is computed.  If $ g(t) $ is the solution to the differential equation 
\begin{align}
 \dot{g} (t)&= {\bf D}_e L_g (\xi (t)), \quad g(0) = e, \label{eq:gdot}
\end{align}
where $L_g: G \rightarrow G$ is left translation by $g$ and ${\bf D}_e L_g$ is its derivative map at the identity element $e$, 
 then the trajectory 
\begin{align}
	c(t)&:= g(t) \bar{c} (t) \label{eq:reconstraj}
\end{align}
solves the reconstruction problem.  That is to say, $ c(t) $ is a solution of \eqref{originalEqMot}.

The objective, however, is not to solve the reconstruction problem. The reduced system on $\mathcal{P}_\mu$ is, in most cases, still nonlinear and finding a solution trajectory $c_\mu(t)$ can prove to be difficult or even impossible. However, restricting to closed trajectories $ c _\mu : [0, 1] \longrightarrow \mathcal{P}_\mu $, the above reconstruction ideas, applied at the initial point $c_{\mu 0}$ (or point of return), naturally lead to the concepts of geometric phase and dynamic phase. Hence the terminology `reconstruction phases'. 

At the end of period 1, both  $\bar{c}(t)$ and $c(t)$ return to the fiber over $c_{\mu 0}$, thus providing the relations
\[\bar{c}(1)=\bar{g}(1) \cdot \bar{c}(0), \quad c(1)=h(1) \cdot c(0), \quad \bar{g}(1),h(1) \in G_\mu \,. \] 
The element $ \bar{g}(1)$ is the \defn{geometric phase} and is computed using~(\ref{eq:gphasegen}) when $G_\mu$ is Abelian. The element $ h(1) $ is the \defn{total phase}. Combining the above and~(\ref{eq:reconstraj}) gives another relation 
\[h(1)=g(1) \, \bar{g}(1), \quad g(1) \in G_\mu \,.\]
The element $g(1)$ is the \defn{dynamic phase}.

The dynamic phase is computed, in principle, by integrating~(\ref{eq:gdot}).  For $G_\mu$ Abelian, the task is somewhat simplified as follows. The equivariance condition~(\ref{eq:equiv1form}) becomes $\phi_g^* \alpha=\alpha$ and the $G_{\mu}$-symmetry of  $X_H$ (${\bf D} \varphi_g  \cdot X_H(p)=X_H(\varphi_g \cdot p)$) implies 
\begin{align}
\xi(t)&:= \alpha (X _H (\bar{c} (t)))=\alpha \left(X_H(c(t))\right) \label{eq:xitab}  \,. 
\end{align} 
Moreover,  the exponential map $\operatorname{exp}: \mathfrak{g}_{\mu} \rightarrow G_{\mu}$ is onto \cite{Ma92} so that at each $t$ there exists $\eta(t) \in \mathfrak{g}_{\mu}$ such that $\operatorname{exp}(\eta(t))=g(t)$. Differentiate the exponential map at $\eta$ with $g$ fixed to obtain the linear map $\mathbf{D}_\eta \operatorname{exp}: T_\eta \mathfrak{g} \equiv \mathfrak{g} \rightarrow T_gG$ and the equation 
 \[\mathbf{D}_\eta \operatorname{exp} (\eta'(t))=\left({\bf D}_{e} L_{g}\right) \xi(t) \,. \] The above is an equation between derivatives of maps and holds irrespective of the base points $(\eta,g)$. In particular, applied at $(0,e)$, one obtains 
\[\frac{d \eta}{dt}=\xi(t)\] (see Warner \cite{Wa83}, for example, for the proof that $\mathbf{D}_\eta \operatorname{exp}$, at $\eta=0$, is the identity). The general formula for the dynamic phase, for Abelian $G_\mu$,  is therefore
\begin{align}
g(1)&= \operatorname{exp}(\eta(1))=\operatorname{exp} \left(\int_0^1 \xi(t) dt \right), \label{eq:dphasegen}
\end{align}
after noting that $g(0)=e$ and $\eta(0)=0$. Obtaining a simple exact expression from the above depends on the form of $\xi(t)$ and this, in turn, depends on the form of the Hamiltonian $H$. It turns out, as will be shown later, that in the three-and-four vortex problems the Hamiltonian $H$ does indeed have a special form which results in a remarkably simple expression for the dynamic phase.

\paragraph{Observation.}  The notation for curves/trajectories introduced in this section is consistently used in the rest of the paper. Namely, the letter $c$ is used, (i) with an appropriate subscript  for curves on symmetry reduced spaces, (ii) with a bar on top for horizontal curves and (iii) with no embellishments for reconstructed trajectories.

\section{Reconstruction phases for the 3-vortex problem.}
\label{sec:phasesThreeVortices}

In this section we describe a sequence of three canonical transformations that implement symplectic reduction of the three vortex problem and proceed to compute the geometric and dynamic phases. It is assumed throughout that $ \Gamma_{\text{\rm tot}} = \Gamma _1 + \Gamma _2 + \Gamma _3 $ (the total vortex strength) is non-zero. 

\subsection{Transformation to Jacobi-Bertrand-Haretu coordinates}
\label{sec:jacobiCoordinates}

Let us introduce the following sequence of canonical transformations.  

Define $ T _1 : \mathbb{C} ^3 \longrightarrow \mathbb{C} ^3 $ by
\[
	(z _1, z _2, z _3) \mapsto (Z _0, r, s)
\]
with
\begin{equation} \label{firstCanTrans}
\begin{split}
	Z _0 &= \frac{ 1 }{ \Gamma_{\text{\rm tot}} } \sum _{ j = 1 } ^3 \Gamma _j \, z_j  \quad \text{(center of circulation)} \,, \\
	r &= z _2 - z _1 \,, \\
	s &= z _3  -  \frac{ \Gamma _1 \, z _1 + \Gamma _2 \, z _2 }{ \Gamma _1 + \Gamma _2 }  \,. 
\end{split}
\end{equation}

It is verified that $ T _1 $ is invertible (in fact, $ \det T _1 = 1 $).  Transformation $ T _1 $ accomplishes reduction to the center of circulation frame by setting $ Z _0 $ at the origin.  Vectors $r$ and $s$ are the \emph{Jacobi-Bertrand-Haretu} (JBH) coordinates of the system.  (See \cite{Sch2002} for an iterative method for constructing JBH coordinates.)

In agreement with \eqref{vortexSymplecticForm}, let $ \Omega _0 $ denotes the symplectic form for the three-vortex problem:
\begin{equation} \label{vortexSymplecticFormReal}
	\Omega _0 := \Gamma _1 \, d x _1 \wedge d y _1 + \Gamma _2 \, d x _2 \wedge d y _2 + \Gamma _3 \, d x _3 \wedge d y _3 \,, 
\end{equation} 
with $ z _k = x _k + \mi y _k $.  The symplectic form after $ T _1 $ becomes
\begin{displaymath}
  \Omega _1 \stackrel{\text{\tiny\rm def}}{=} T _1 ^\ast \Omega _0 = \Gamma_{\text{\rm tot}} \, d Z _{ 0x } \wedge d Z _{ 0y } + A  \, d r _x \wedge d r _y + B \, d s _x \wedge d s _y \,,
\end{displaymath}
where
\begin{equation} \label{jacobiVectorsCoefficients}
	A := \frac{ \Gamma _1 \Gamma _2 }{ \Gamma _1 + \Gamma _2 }\,, \quad
	B := \frac{ (\Gamma _1 + \Gamma _2) \Gamma _3 }{ \Gamma _1 + \Gamma _2 + \Gamma _3 } \,,
\end{equation}
and the subindices $x$ and $y$ indicate real and imaginary parts.

\begin{remark}
Apart from the coefficients $ \Gamma_{\text{\rm tot}}, A, B $, which depend on the vortex strengths, $ \Omega _1 $ retains the basic structure of $ \Omega _0 $, namely that conjugate coordinates are the real and imaginary parts of each complex coordinate.  This is the distinctive feature of using JBH coordinates.
\end{remark}

\begin{remark}
If $ \Gamma_{\text{\rm tot}} = 0 $ then $ Z _0 $ is not defined and thus we do not have JBH coordinates.  This is one reason for which we restrict ourselves to the case of non-zero total vortex strength.  Apart from this, it should be noted that when $ \Gamma_{\text{\rm tot}} = 0 $ the coadjoint isotropy subgroup is the full $\operatorname{SE}(2)$ group, which is not Abelian, and hence formula \eqref{eq:gphasegen} can no longer be applied.
\end{remark}
  
\subsection{Transformation to action-angle coordinates}
\label{sec:AACoords}

Define $ T _2 : \mathbb{C} ^3 \longrightarrow \mathbb{R}^2 \times \mathbb{R} ^2 \times \mathbb{T} ^2 $ by
\[
	(Z _0, r, s) \mapsto (K _x, K _y \,;\; j _1, j _2 \,;\; \theta _1, \theta _2)
\]
with
\[
	Z _0 = {\textstyle\frac1{\sqrt{\Gamma_{\text{\rm tot}}}}} \left( K _x + \mi K _y \right)
\]
and
\begin{equation} \label{jacobiVectors}
	r = \frac{ \sqrt{2 j _1} \, e ^{ \mi \theta _1 } }{ \sqrt{A} }
	\,, \quad
	s = \frac{ \sqrt{ 2 j _2} \, e ^{ \mi \theta _2 }}{ \sqrt{B} }
\end{equation}
with $A$ and $B$ as in \eqref{jacobiVectorsCoefficients}.  Hence
\begin{equation} \label{signsj1j2}
	\sign(j _1) = \sign(A) \quad \text{and} \quad \sign(j _2) = \sign(B) \,.
\end{equation}

Transformation $ T _2 $ is akin to passage to polar coordinates of $r$ and $s$.  
The coefficients $ 1/\sqrt{A} $ and $ 1/\sqrt{B} $ accomplish the task of cancelling-out the coefficients accompanying the terms in $ \Omega _1 $, so that we get the standard symplectic form:
\[
  \Omega _2 \stackrel{\text{\tiny\rm def}}{=} T _2 ^\ast \Omega _1 = d K _x \wedge d K _y + d j _1 \wedge d \theta _1 + d j _2 \wedge d \theta _2 \,. 
\]
Reduction by translational symmetry is achieved by setting the center of circulation at the origin: $ K _x = K _y = 0 $.

\subsection{Reduction by rotational symmetry}
\label{sec:rotationalSymmetry}

After reducing by translational symmetry ($ K _x = K _y = 0 $), the hamiltonian depends on $ \theta _1 $ and $\theta _2 $ only through their difference.  In order to have the (semi-) sum and difference of $ \theta _1 $ and $ \theta _2 $ as coordinates, we introduce a third canonical transformation $ T _3 : \mathbb{R}^2 \times \mathbb{R} ^2 \times \mathbb{T} ^2 \longrightarrow \mathbb{R}^2 \times \mathbb{R} ^2 \times \mathbb{T} ^2 $ by
\[
  (K _x, K _y \,;\; j _1, j _2 \,;\; \theta _1, \theta _2) \mapsto (\tilde{K} _x, \tilde{K} _y \,;\; I _1, I _2 \,;\; \varphi _1, \varphi _2)
\]
through the (type II) generating function
\[
  G _2 = K _x \tilde{K} _y + j _1 (\varphi _2 - \varphi _1) + j _2 (\varphi _1 + \varphi _2)
\]
with the relations
\[
  \theta _k = \frac{ \partial G _2 }{ \partial j _k } \,, \quad I _k = \frac{ \partial G _2 }{ \partial \varphi_k } \,, \quad k = 1, 2
\]
and $ K _y = \partial G _2 / \partial K _x $, $ \tilde{K} _x = \partial G _2 / \partial \tilde{K} _y $.  That is to say,
\[
	(\tilde{K} _x, \tilde{K} _y) = (K _x, K _y) \,,
\]
and
\begin{align}
  \label{coordsT3I1I2}
  I _1 &=  j _2 - j _1 \,,  \hspace{-3em}&\hspace{-3em}  I _2 &= j _1 + j _2 \,, \\
  \label{coordsT3Phi1Phi2}
  \varphi _1 &= \frac{ \theta _2 - \theta _1 }{ 2 } \,,  \hspace{-3em}&\hspace{-3em}  \varphi _2 &= \frac{ \theta _1 + \theta _2 }{ 2 } \,. 
\end{align}
Note that, being half the angle between vectors $r$ and $s$, $ \varphi _1 \in [0, \pi] $.

Since $ T _3 $ is defined through a generating function, 
$ \Omega _3 := T _3 ^\ast \Omega _2 $ is again the standard symplectic form
\[
  \Omega _3 = d \tilde{K} _x \wedge d \tilde{K} _y + d I _1 \wedge d \varphi _1 + d I _2 \wedge d \varphi _2 \,. 
\]

Let $ P _0 $ represent all vortex configurations (including collisions) for which the center of circulation is at the origin:
\[
  P _0 \stackrel{\text{\tiny def}}{=} \{ (z _1, z _2, z _3) \in \mathbb{C} ^3 \mid \sum \Gamma _i z _i = 0 \} \,. 
\]
Then $ P _0 $ is parametrized by the canonical coordinates $ ( I _1, I _2, \varphi _1,\varphi _2 ) $, with $ (\tilde{K} _x, \tilde{K} _y) = (0, 0) $.

Note that $ \varphi _2 $ keeps track of rigid rotations of the vortex configuration.  In other words, $ SO(2) $ acts on $ P _0 $ through the rule
\begin{equation} \label{thetaAdditiveAction}
  \theta \cdot (I _1, I _2, \varphi _1, \varphi _2) = (I _1, I _2, \varphi _1, \varphi _2 + \theta) \,. 
\end{equation}
Thus, at every $ p \in P _0 $, the vertical space is
\[
  V _p = \operatorname{span} \left( \frac{ \partial }{ \partial\varphi _2 } \right) \,. 
\]

Since the Hamiltonian is cyclic in $ \varphi _2 $, then $ I _2 $ is a constant of motion.  Moreover,
\[
	\mathbf{i} _{ \partial / \partial \varphi _2 } \Omega _3 = - d I _2 \,, 
\]  and so the real-valued function defined as $\mathbf{J}(I _1, I _2, \varphi _1, \varphi _2)=-I _2 $ is the momentum map of the $\operatorname{SO}(2)$ action.  Since $ \partial / \partial \varphi _2 = \partial  \left( e ^{ i \theta } z _1, e ^{ i \theta } z _2, e ^{ i \theta } z _3 \right) / \partial \theta = \sum _k (-y _k \, \partial / \partial x _k + x _k \, \partial / \partial y _k ) $, 
it follows that \begin{align}
 I _2 &= \frac{\Theta_0}{2}= -\mathbf{J}(I _1, I _2, \varphi _1, \varphi _2)= - \mu \,,  \label{eq:I2etc}
\end{align}
where $\Theta_0$ is the second moment of circulation~(\ref{eq:circonstants}). The symbol $\mu$, introduced earlier, is adopted henceforth for the image of $\mathbf{J}$ since it is standard notation. Throughout we will assume $ \mu \neq 0 $.

With reference to~(\ref{JmuPmuGmu}), consider now the bundle $ \pi _\mu : \mathbf{J} ^{-1} (\mu) \longrightarrow \mathcal{P} _\mu $, where the symplectic reduced space
\[
  \mathcal{P} _\mu \stackrel{\text{\tiny def}}{=} \mathbf{J}^{-1} (\mu) \, \big{/} \, SO(2)
\]
is parametrized by $ (I _1, \varphi _1) $.  Notice that $ \mu ^2 - I _1 ^2 = 4 j _1 j _2 $, so that whether $ I _1 $ remains bounded depends on $ \sign( j _1 j _2) $; we elaborate on this below.  Notice also that, regardless of the value of $ \varphi _1 $, $ I _1 = - \mu $ ($I _1 = + \mu $) corresponds to $ j _1 = 0 $ ($ j _2 = 0 $), which in turn corresponds to a binary collision of vortices $ \Gamma _1 $ and $ \Gamma _2 $ (a collinear configuration with $ \Gamma _3 $ at the center of circulation of $ \Gamma _1 $ and $ \Gamma _2 $).

\begin{figure}
\includegraphics[scale=0.37]{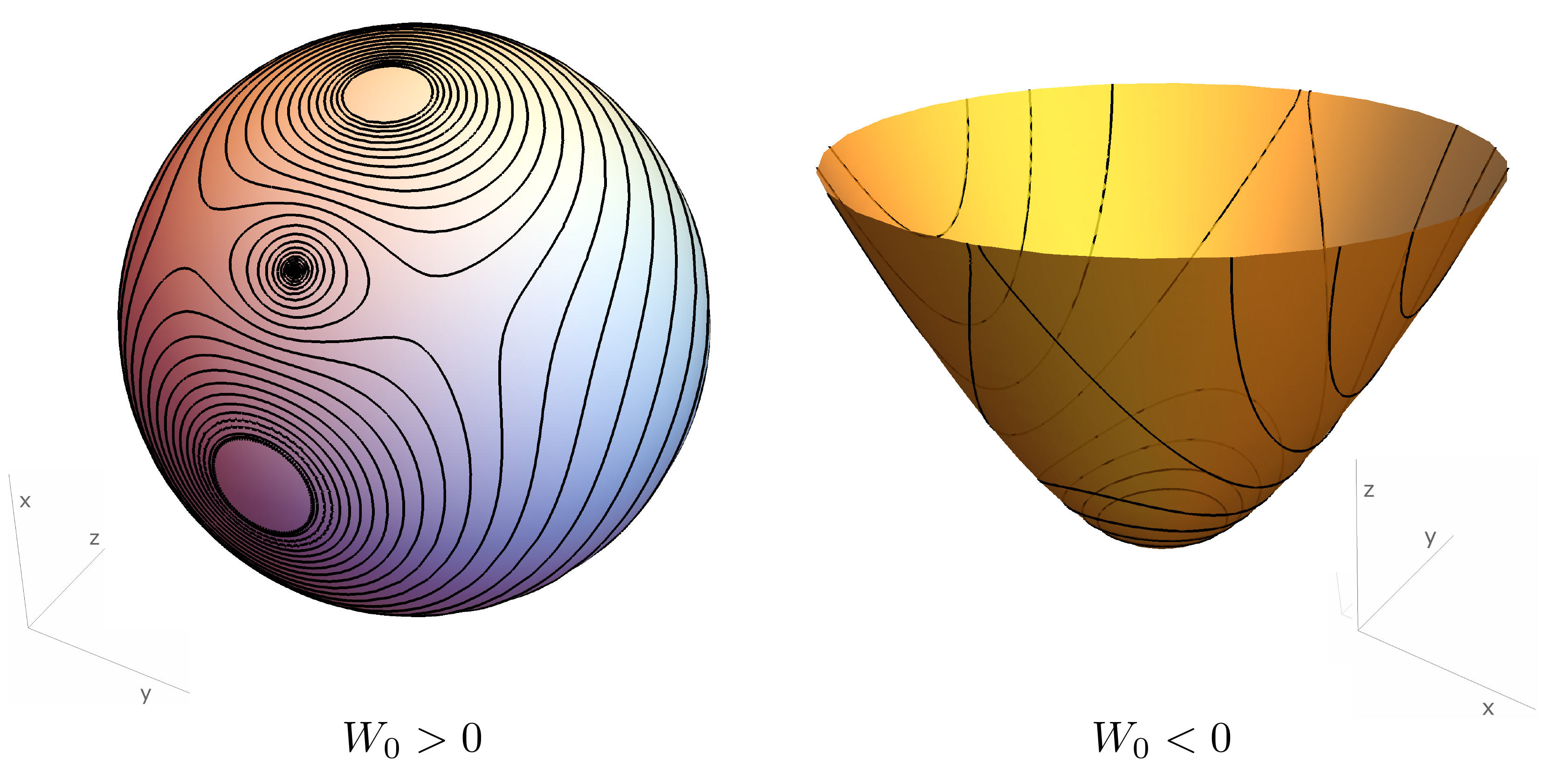}
\caption{\label{fig:reducedspaces}Symplectic reduced spaces $ \mathcal{P} _\mu $ for (a) $ W _0 > 0 $ and (b) $ W _0 < 0 $.    For both cases $ | \mu | = 1 $.  Flow lines correspond to (a) $ \Gamma _1 = 7.615 $, $ \Gamma _2 = -3.46 $, $ \Gamma _3 = -3.155 $ and (b) $ \Gamma _1 = -1 $, $ \Gamma _2 = -1 $, $ \Gamma _3 = 1 $.}
\end{figure}

Let us now look at how the boundedness of the dynamics depends on the vortex strengths.  From \eqref{jacobiVectorsCoefficients} and \eqref{signsj1j2},
\[
	\sign(j _1 j _2) = \sign(AB) = \sign \left( \frac{ \Gamma _1 \Gamma _2 \Gamma _3 }{ \Gamma _1 + \Gamma _2 + \Gamma _3 } \right) = \sign (W _0)
\] 
where 
\[
	W _0 \stackrel{\text{def}}{=} \frac1{\Gamma _1 \Gamma _2} + \frac1{\Gamma _2 \Gamma _3} + \frac1{\Gamma _3 \Gamma _1} \,. 
\] 
Therefore, $ \mathcal{P} _\mu $ is bounded precisely when $ W _0 > 0 $.  In this case $ \mathcal{P} _\mu $ is identified with $ S ^2 _\mu $, the sphere of radius $ |\mu| $
\begin{equation} \label{spheremu}
	S ^2 _\mu : \quad  x ^2 + y ^2 + z ^2 = \mu ^2 \,, 
\end{equation}
parametrized by
\begin{equation}\label{cyl_coords_sphere}
	(x,y,z) = \left( \sqrt{ \mu ^2 - I _1 ^2 } \, \cos 2 \varphi _1, \sqrt{ \mu ^2 - I _1 ^2 } \, \sin 2 \varphi _1, \; I _1 \right) \,. 
\end{equation}
If $ W _0 < 0 $ then $ \mathcal{P}_\mu $ is identified with the two-sheeted hyperboloid $ H _\mu $ given by
\begin{equation} \label{hyperboloidmu}
	H ^2 _\mu : \quad  z ^2 - x ^2 - y ^2 = \mu ^2 \,, 
\end{equation}
parametrized by
\[
	(x,y,z) = \left( \sqrt{ I _1 ^2 - \mu ^2 } \, \cos 2 \varphi _1, \sqrt{ I _1 ^2 - \mu ^2 } \, \sin 2 \varphi _1, \; I _1 \right) \,. 
\]
In both cases $ (I _1, 2 \varphi _1) $ are \emph{cylindrical coordinates} of $ \mathcal{P} _\mu $.  We will refer to $ \{ (0, 0, -\mu), (0, 0, \mu) \} \subset \mathcal{P} _\mu $ as the \emph{singular points} of the coordinate chart.

Figure \ref{fig:reducedspaces} shows the hamiltonian flow lines on the symplectic reduced space $ P _\mu $ for particular choices of the vortex strengths corresponding to $ W _0 > 0 $ and $ W _0 < 0 $.

Let us refine the description of $ \mathcal{P} _\mu $ in the non-compact case $ W _0 < 0 $.  There are two subcases:  either (i) $ A < 0 $ and $ B > 0 $ or (ii) $ A > 0 $ and $ B < 0 $.  Assume that $ \mu > 0 $.  It is easy to see that the dynamics takes place in the upper or lower sheet of the hyperboloid depending on whether (i) or (ii) holds.  If $ \mu < 0 $ then the correspondence is reversed.  Thus, the reduced space $ P _\mu $ is actually one of the two sheets of the hyperboloid.

\begin{remark}
	The classification of the dynamics in the compact and non-compact cases is discussed in \cite{BoLe98}, using $ W _0 $ as classification parameter.  Interestingly, $ \mathcal{P} _\mu $ is obtained as a coadjoint orbit in $ \mathfrak{u} (2) ^\ast $.  This is shown using the Poisson structure of the 3-vortex problem expressed in terms of symmetry invariants.  A further account of this point of view is found in \cite{He2016a}.
\end{remark}

\subsection{Metric and connection}

The next step is to define a connection on the bundle $ \pi _\mu : \mathbf{J} ^{-1} (\mu) \longrightarrow \mathcal{P} _\mu $.  Consider the Hermitian metric on $ \mathbb{C} ^N $ defined by
\begin{align}
	\big\langle U, V \big\rangle(p)&:= \sum _{ k = 1 } ^n \Gamma _k \, u _k \, \bar{v}_k,  \label{eq:hermitianMetric}
\end{align}
where $p \equiv (z_1, \ldots, z_n) \in \mathbb{C} ^N$, $U \equiv (u_1, \ldots, u_n) \in T_p \mathbb{C} ^N $ and $V \equiv (v_1, \ldots, v_n) \in T_p \mathbb{C} ^N $
with the associated Riemannian metric and symplectic form:
\begin{equation} \label{eq:metricAndSymplecticForm}
		\left\langle \! \left\langle U, V \right\rangle \! \right\rangle \stackrel{\text{\tiny\rm def}}{=} \Re \langle U, V \rangle \,, \quad \Omega _0 (U, V) = -\Im \langle U, V \rangle  \,. 
\end{equation}
Note that this Hermitian metric makes $ \mathbb{C} ^N $ a Hermitian manifold with respect to the complex structure associated to the symplectic form.
Also, note that the symplectic form just defined coincides with \eqref{vortexSymplecticForm}, the symplectic form for the $N$-vortex problem.  It is easy to see that $ \left\langle \! \left\langle \;,\, \right\rangle \! \right\rangle $ is $G$-invariant and that the \emph{angular impulse} takes the form
\[
	\Theta_0(\mathbf{z}) = \frac12 \left\langle \! \left\langle \mathbf{z}, \mathbf{z} \right\rangle \! \right\rangle \,, \quad \mathbf{z} \in \mathbb{C} ^N \,. 
\] 

\begin{remark}
	If $ W _0 < 0 $ then the metric $ \left\langle \! \left\langle \;,\, \right\rangle \! \right\rangle $ is indefinite.  That this is so can be seen from the expression for the metric in $ (r, s) $ coordinates computed below.  This, however, plays no role in the construction of the horizontal bundle and the connection 1-form because, provided that the angular impulse is not zero, the vertical space is definite.
\end{remark}

Let us return to the case $ N = 3 $.  The pullback of the metric to $ \mathbf{J} ^{-1} (\mu) $ induces a connection on $ \pi _\mu : \mathbf{J} ^{-1} (\mu) \longrightarrow \mathcal{P} _\mu $ by identifying the horizontal bundle with the metric-orthogonal to the vertical bundle:
\[
	H _p = V _p ^\perp \stackrel{\text{\tiny\rm def}}{=} \{ w _p \in T _p \mathbf{J} ^{-1} (\mu) \mid \left\langle \! \left\langle w _p, v _p \right\rangle \! \right\rangle = 0 \} \,, 
\] 
where $ v _p = \operatorname{span} \left( \xi _{ J ^{-1} (\mu) } (p) \right) $, for any $ \xi \in \mathfrak{so}(2) $, $ \xi \neq 0 $.

The level set $ \mathbf{J} ^{-1} (\mu) \subset P _0 $ is parametrized by coordinates $ (I _1, \varphi _1, \varphi _2) $, keeping $ I _2 = -\mu $ fixed and $ (\tilde{K} _x, \tilde{K} _y) = (0,0) $.  Also, the action of $ SO(2) $ only affects $ \varphi _2 $, as stated in \eqref{thetaAdditiveAction}.  Thus, the vertical operator associated to the metric-orthogonal connection is given by the following:

\vspace{1ex}
\begin{lemma}
	Let $ \pi : E \longrightarrow B $ be a principal fiber bundle with typical fiber diffeomorphic to the structural group $G$, of dimension one, and where $E$ is a Riemannian manifold.  Suppose that $ (q ^1, \ldots, q ^n) $ are coordinates on $E$ such that $ V _p = \operatorname{span} \left( \partial / \partial q ^n \right) $.  Equip the fiber bundle with the metric-orthogonal connection, i.e. $ H _p = V _p ^\perp $.  Then the connection 1-form is
	\[
	  \alpha = \frac1{g _{nn}} \sum _{ k = 1 } ^{ n-1 } g _{ nk } \, d q ^k + d q ^n \,,
	\] 
	where the $ g _{ ij } $ are the coefficients of the metric in coordinates $ (q ^1, \ldots q ^n) $.
\end{lemma}
\emph{Proof:}  We need to verify the properties~(\ref{ConnectionProperties}). Assuming that the action $g \cdot p$, where $g \in G$, is independent of $p \in E$, it is sufficient to verify that \emph{i)} $ \alpha (v) = v $ for all $ v \in V _p $ and \emph{ii)} $ \alpha (w) = 0 $ for all $ w \in H _p $.  The first claim is obvious.  As for the second, suppose that $ w \in H _p $, i.e. $ w = w ^i \partial / \partial q ^i $ with $ \sum _{ k = 1 } ^n g _{ nk } w ^k = 0 $.  Then
\[
	\qquad
\alpha (w) = \frac1{g _{ nn }} \left( \sum _{ k = 1 } ^{ n-1 } g _{ nk } w ^k + g _{ nn } w ^n \right) = \frac1{g _{ nn }} \sum _{ k = 1 } ^n g _{ nk } w ^k = 0 \,.
	\qquad \square
\]

\vspace{1ex}
Applied to the bundle $ \pi _\mu : \mathbf{J} ^{-1} (\mu) \longrightarrow \mathcal{P} _\mu $ the connection 1-form has the expression:
\begin{equation} \label{verticalOperatorEx1}
	\alpha = \frac1{ \scriptstyle \left\langle \! \left\langle \frac{ \partial }{ \partial \varphi _2 }, \frac{ \partial }{ \partial \varphi _2 } \right\rangle \! \right\rangle }
	\left[
		\textstyle
		\left\langle \! \left\langle \frac{ \partial }{ \partial I _1 }, \frac{ \partial }{ \partial \varphi _2 } \right\rangle \! \right\rangle d I _1 
		+ \left\langle \! \left\langle \frac{ \partial }{ \partial \varphi _1 }, \frac{ \partial }{ \partial \varphi _2 } \right\rangle \! \right\rangle d \varphi _1 
	\right]
	+ d \varphi _2 \,. 
\end{equation}
To compute the metric coefficients in coordinates $ (I _i, \varphi _i) $ let us first compute the metric coefficients in coordinates $ (j _i, \theta _i) $.
Let $ \mathbb{T} _1 $ be the matrix representation (with respect to the canonical basis) of the linear transformation $ T _1 : \mathbb{C} ^3 \longrightarrow \mathbb{C} ^3 $ introduced in subsection \ref{sec:jacobiCoordinates}.  A straightforward computation shows that
\[
	\begin{pmatrix} \Gamma _1 &&\\ & \Gamma _2 &\\ && \Gamma _3 \end{pmatrix}
	= \mathbb{T} _1 ^T \begin{pmatrix} \Gamma_{\text{\rm tot}} &&\\ & A &\\ && B \end{pmatrix} \mathbb{T} _1 \,, 
\] 
with $A$ and $B$ defined in \eqref{jacobiVectorsCoefficients}.  It follows that, in coordinates $ (r,s) \in \mathbb{C} ^2 $, the matrix representation of the Hermitian metric restricted to $ P _0 $ is $ \begin{pmatrix} A &\\ & B \end{pmatrix} $.
Moreover, from \eqref{jacobiVectors},
\begin{align*}
	\frac{ \partial }{ \partial j _1 } &= \left( \frac{ \me ^{ \mi \theta _1 } }{ \sqrt{A} \, \sqrt{2 j _1} }, 0 \right) \,, &
	\frac{ \partial }{ \partial j _2 } &= \left( 0, \frac{ \me ^{ \mi \theta _2 } }{ \sqrt{B} \, \sqrt{2 j _2} } \right) \,, \\[1ex]
	\frac{ \partial }{ \partial \theta _1 } &= \left( \frac{ \mi \, \sqrt{2 j _1} \, \me ^{ \mi \theta _1 } }{ \sqrt{A} }, 0 \right) \,, &
	\frac{ \partial }{ \partial \theta _2 } &= \left( 0, \frac{ \mi \, \sqrt{2 j _2} \, \me ^{ \mi \theta _2 } }{ \sqrt{B} } \right) \,. 
\end{align*}
Now, from \eqref{coordsT3I1I2} and \eqref{coordsT3Phi1Phi2}, 	it is easy to compute:
\[\begin{split}
	\left\langle \frac{ \partial }{ \partial I _1 }, \frac{ \partial }{ \partial \varphi _2 } \right\rangle &= \left\langle -\frac12 \frac{ \partial }{ \partial j _1 } + \frac12 \frac{ \partial }{ \partial j _2 }, \frac{ \partial }{ \partial \theta _1 } + \frac{ \partial }{ \partial \theta _2 } \right\rangle = 0 \,, \\
	\left\langle \frac{ \partial }{ \partial \varphi _1 }, \frac{ \partial }{ \partial \varphi _2 } \right\rangle &= \left\langle - \frac{ \partial }{ \partial \theta _1 } + \frac{ \partial }{ \partial \theta _2 }, \frac{ \partial }{ \partial \theta _1 } + \frac{ \partial }{ \partial \theta _2 } \right\rangle = 2 (j _2 - j _1) = 2 I _1 \,, \\
	\left\langle \frac{ \partial }{ \partial \varphi _2 }, \frac{ \partial }{ \partial \varphi _2 } \right\rangle &= \left\langle \frac{ \partial }{ \partial \theta _1 } + \frac{ \partial }{ \partial \theta _2 }, \frac{ \partial }{ \partial \theta _1 } + \frac{ \partial }{ \partial \theta _2 } \right\rangle = 2 (j _1 + j _2) = 2 I _2 \,.
\end{split}\] 
Note that these three expressions are real, and so they give the metric coefficients in coordinates $ I _1, \varphi _1, \varphi _2 $.  Thus, expression \eqref{verticalOperatorEx1} for the vertical operator takes the form $ \operatorname{ver} = (I _1 / I _2) d \varphi _1 + d \varphi _2 $.  Hence, the connection one-form on $\mathbf{J}^{-1} (\mu)$, using~(\ref{eq:I2etc}), is
\begin{equation} \label{conn1form3v}
	\alpha = -\frac{I _1}{ \mu } d \varphi _1 + d \varphi _2 \,.
\end{equation}

\subsection{The geometric phase}

In this subsection we proceed to compute the geometric phase for the three-vortex problem.

\subsubsection*{Geometric phase as a line integral}

Let $ [0, 1] \ni t \mapsto c_\mu(t) = \left( I _1(t), \varphi _1(t) \right) $ be a closed loop in $ \mathcal{P} _\mu $ and assume it does not pass through a singular point of the coordinate chart.  A section $ \sigma : \mathcal{P} _\mu \longrightarrow J ^{-1} (\mu) $ is obtained by fixing the direction of either vector $r$ or $s$ (defined in \eqref{firstCanTrans}), that is to say, by keeping either $ \theta_1 $ or $ \theta _2 $ constant.  From \eqref{coordsT3Phi1Phi2} and \eqref{conn1form3v} we get that the connection 1-form, expressed in $ (I _1, I _2, \varphi _1, \theta _1) $ coordinates, is
\[
	\alpha = \left( 1 - \frac{ I _1 }{ \mu } \right) d \varphi _1 + d \theta _1 
\] 
and thus its pull-back by the section is
\[
	\sigma ^\ast \alpha = \left( 1 - \frac{ I _1 }{ \mu } \right) d \varphi _1 \,. 
\] 
The curvature of $\alpha$ is the 2-form
\[
	D \alpha := d \sigma ^\ast \alpha = - \frac1{ \mu } d I _1 \wedge d \varphi _1 \,. 
\] 
Hence, the geometric phase associated to a horizontal lift $ \bar c(t) = (I _1 (t), \varphi _1 (t)$, $\varphi_2 (t)) $ over $ c _\mu (t) $, computed from \eqref{eq:gphasegen}, is given by
\begin{equation} \label{geometricPhaseDefined}
	\theta _g \stackrel{\text{\tiny\rm def}}{=} - \oint _{ c _\mu } \sigma ^\ast \alpha = - \int _0 ^1 \left( 1 - \frac{ I _1 (t) }{ \mu } \right) \dot{\varphi} _1 (t) \, d t \,, 
\end{equation}
after identifying $\operatorname{exp}: \mathbb{R} \rightarrow \operatorname{SO}(2)$ with the $\operatorname{mod} 2 \pi$ map.

\subsubsection*{Geometric phase as a double integral}

Let us now express the geometric phase in terms of a double integral.  We begin by endowing $ P _\mu $ with an area element:

\begin{definition} \label{def:omegamu}
Let $ \omega _\mu $ be the 2-form on the reduced space $ \mathcal{P}_\mu $ given by
\[
	\omega _\mu := 2 \mu \, d I _1 \wedge d \varphi _1 \,, \quad \mu \neq 0 \,. 
\] 
Give $ \mathcal{P}_\mu $ the orientation induced by $ \omega_\mu $.
\end{definition}

\begin{remark}
Let us comment on the geometric meaning of $ \omega _\mu $.  Endow $ \mathbb{R} ^3 $ with the standard metric $ d s ^2 = d x ^2 + d y ^2 + d z ^2 $ if $ W _0 > 0 $ (the compact case) or the \emph{hyperbolic metric} $ d s ^2 = d z ^2 - d x ^2 - d y ^2 $ if $ W _0 < 0 $ (the non-compact case).  Then, for both cases,
\[
	\mathcal{P} _\mu = \{ (x,y,z) \in \mathbb{R} ^3 \mid \| (x,y, z) \| ^2 = \mu ^2 \} \,, \quad \mu \neq 0 
\] 
with the norm $ \| \; \| $ induced  by the corresponding metric.  It is easy to verify that the area element on $ \mathcal{P} _\mu $ induced by $ d s ^2 $ is, in both cases,
\[
	d A = 2 \, |\mu| \, d z \wedge d \varphi _1 
\] 
with coordinates $ (z, \varphi _1) $ defined according to \eqref{spheremu}--\eqref{hyperboloidmu}.  Thus, if $ \mu > 0 $ and provided that we choose the appropriate metric (or pseudo-metric) on $ \mathbb{R} ^3 $ depending on $ \sign W _0 $, the 2-form $ \omega _\mu $ is the area element of $ P _\mu $ as an embedded surface on $ \mathbb{R} ^3 $.  For a general $ \mu \neq 0 $, the orientation induced by $ \omega _\mu $ depends on $ \sign \mu $.
\end{remark}

Consider a simple closed trajectory $ c _\mu (t) $ on $ \mathcal{P} _\mu $.  Assume initially that $ W _0 > 0 $.  From the two regions having $ c _\mu $ as boundary, let $U$ be the one with respect to which $ c _\mu (t) $ is positively oriented.  Then, by Stokes' theorem,
\begin{align*}
	\theta _g &= - \sign(\mu) \iint _U D \alpha = \frac{ \sign(\mu) }{ \mu } \iint _U d I _1 \wedge d\varphi _1  \\
	&= \frac{1}{2 \mu ^2} \iint _U 2 | \mu | \, d I _1 \wedge d \varphi _1 = \frac{1}{ 2 \mu ^2 }\iint _U \omega _{ |\mu| } 
\end{align*}
where the factor $ \sign(\mu) $ takes care of the fact that the orientations induced by $ \omega _\mu $ and $ d I _1 \wedge d \varphi _1 $ are opposite if $ \mu < 0 $.  To further simplify this expression, let $ U _1 $ be te radial projection of $U$ on the unit sphere.  Since $ \omega _{ | \mu | } = \mu ^2 \omega _1 $, where $ \omega _1 $ is the standard area element on the unit sphere, we obtain
\begin{equation} \label{geometricPhaseSimpleArea}
	\theta _g = \frac{ 1 }{ 2 }  \iint _{ U _1 }  \omega _1 \qquad (\operatorname{mod} 2 \pi) \,. 
\end{equation}
In this simplified formula we have added that the result is to be taken modulo addition of an integer multiple of $ 2 \pi $, since $ \theta _g $ is an angle.

If $ W _0 < 0 $ then there is just one region $U$ whose boundary is $ c _\mu $.  In this case assume that $ c _\mu (t) $ is positively oriented with respect to $U$.  The argument leading to formula \eqref{geometricPhaseSimpleArea} goes through as before, noticing that now $ U _1 $ is the projection of $U$ on the normalized two-sheeted hyperboloid $ H ^2 _1 $ and $ \omega _1 $ is its area form according to definition \ref{def:omegamu}.

In summary:
\begin{quote}
	The geometric phase associated to a closed loop bounding a region $U$ on $ \mathcal{P} _\mu $, positively oriented, equals half the area projected by $U$ on the concentrical unit sphere (for the compact case) or on the normalized two-sheeted hyperboloid (for the non-compact case).
\end{quote}

\begin{remark}
	The reader might object that the application of Stoke's theorem requires that $ c _\mu $ be homotopic to a point within the chart given by coordinates $ (I _1, \varphi _1) $.  For example, when $ \mathcal{P}_\mu $ is a sphere, this fails to be the case if $ c _\mu $ encloses the ``north pole'' of $ \mathcal{P}_\mu $.  In proposition \ref{prop:geometric_phase_JBH_charts} we prove that formula \eqref{geometricPhaseSimpleArea} holds even in this case.  Indeed, the above formula for the geometric phase uses the Jacobi-Bertrand-Haretu coordinates and the canonical transformation $T_1$ as defined by relations~(\ref{firstCanTrans}). However, there is nothing sacrosanct about the sequence in which the vortices were considered, i.e.  first defining the vector from vortex \emph{one} to vortex \emph{two}, and then the vector from their center of circulation to vortex \emph{three}. In Appendix \ref{apdx:JBH_charts} it is shown that the geometric phase formula is essentially unaltered by considering the other two sequences.
\end{remark}

\subsection{The dynamic phase}

The computation of the dynamic phase is presented in this subsection.  The expression for the dynamic phase turns out to be surprisingly simple and the reason for this is rather interesting.

We begin by recalling Leonard Euler's remarkable theorem for homogeneous functions:

\begin{theorem}[Euler 1755]
	If $ f : \mathbb{R} ^n \setminus \{0\} \longrightarrow \mathbb{R} $ is homogeneous of degree $ \alpha \in \mathbb{R} $ (i.e. $ f(\lambda \mathbf{q}) = \lambda ^\alpha f(\mathbf{q}) $, $\mathbf{q} \in \mathbb{R}^n$, for all $ \lambda > 0 $) then
\[
	\mathbf{q} \cdot \nabla f(\mathbf{q}) = \alpha \, f(\mathbf{q}) \,.
\]
\end{theorem}
\pagebreak[3]
\begin{proof}
Assuming $f$ homogeneous of degree $\alpha$,
\[
	\mathbf{q} \cdot \nabla f(\mathbf{q}) = \left. \frac{d}{dt} \right|_{t=0} f \big( (1+t) \mathbf{q} \big) = \left. \frac{d}{dt} \right|_{t=0} (1 + t) ^\alpha f(\mathbf{q}) = \alpha f(\mathbf{q}) \,.
\]
\end{proof}

What will be relevant to us is the

\begin{corollary} \label{cor:Euler}
	If $f$ is homogeneous of degree 1 then $ \mathbf{q} \cdot \nabla \ln(f(\mathbf{q})) = 1 $.
\end{corollary}
\begin{proof}
	$ \mathbf{q} \cdot \nabla \ln (f(\mathbf{q})) = \left( \mathbf{q} \cdot \nabla f(\mathbf{q}) \right) / f(\mathbf{q}) = 1 $.
\end{proof}

From \eqref{firstCanTrans} it is clear that, given $ \lambda > 0 $, the rescaling $$ (Z _0, r, s) \mapsto (\lambda ^{ 1/2 } Z _0, \lambda ^{ 1/2 } r, \lambda ^{ 1/2 } s) $$ implies the rescaling $ z _i \mapsto \lambda ^{ 1/2 } z _i $, $ i = 1,2,3 $, which in turn implies the rescaling
\begin{equation} \label{sidesSqRescaling}
	b _i \mapsto \lambda b _i \,, \quad i = 1, 2, 3 \,.
\end{equation}
(Recall that $ b _i = | z _j - z _k | ^2 $, $ (i,j,k) \in C _3 $.)  In particular, if the center of circulation $ Z _0 $ is pinned at the origin, the rescaling $ (r, s) \mapsto (\lambda ^{ 1/2 } r, \lambda ^{ 1/2 } s) $ implies \eqref{sidesSqRescaling}.  Furthermore, from \eqref{jacobiVectors} and \eqref{coordsT3I1I2}, we see that the rescaling $ (I _1, I _2) \mapsto (\lambda I _1, \lambda I _2) $ implies \eqref{sidesSqRescaling} as well.  In other words,

\begin{proposition} \label{prop:sidesSqRescaling}
	With the center of circulation kept at the origin, $ b _i $ is a homogeneous function of the variables $ I _1, I _2 $ of degree one, for $ i = 1, 2, 3 $.
\end{proposition}

From corollary \ref{cor:Euler} and proposition \ref{prop:sidesSqRescaling} we obtain at once

\begin{proposition} \label{prop:vortexEuler}
	If the center of circulation is at the origin, then
	\[
	  I _1 \frac{ \partial \ln b _i }{ \partial I _1 } + I _2 \frac{ \partial \ln b _i }{ \partial  I _2 } = 1 \,, \quad i = 1, 2, 3 \,. 
	\] 
\end{proposition}
In terms of coordinates $ (I _1, I _2, \varphi _1, \varphi _2) $, the hamiltonian takes the form
\[
	h = - \frac1{4 \pi} \sum _{ i = 1 } ^3 \Gamma _j \Gamma _k \ln b _i (I _1, I _2, \varphi _1) \,, \quad (i,j,k) \in C _3 \,. 
\] 
Its Hamiltonian vector field is
\[
	X _h = \frac{ \partial h }{ \partial \varphi _1 } \frac{ \partial }{ \partial I _1 } + \frac1{4\pi} \sum_{ (i,j,k) \in C _3 } \Gamma _j \Gamma _k \left( \frac{ \partial \ln b _i }{ \partial I _1 }  \frac{ \partial }{ \partial \varphi _1 } + \frac{ \partial \ln b _i }{ \partial I _2 } \frac{ \partial }{ \partial \varphi _2 } \right) \,.
\]
Hence, using \eqref{conn1form3v} and \eqref{eq:I2etc},
\begin{align*}
	\alpha (X _h) &=  \frac1{I _2} (I _1 d \varphi _1 + I _2 d \varphi_2) (X _h) \\
	&=  \frac1{4 \pi I _2} \sum _{ (i,j,k) \in C _3 } \Gamma _j \Gamma _k \left( I _1 \frac{ \partial \ln b _i }{ \partial I _1 } + I _2 \frac{ \partial \ln b _i }{ \partial I _2 } \right) \\
	&=  -\frac1{4 \pi \mu} ( \Gamma _1 \Gamma _2 + \Gamma _2 \Gamma _3 + \Gamma _3 \Gamma _1 )
\end{align*}
where the last equality follows from proposition \ref{prop:vortexEuler}.  Thus,

\begin{theorem} \label{thm:dynamicPhase}
	Let $c (t) $ be the orbit of the system in $\mathbf{J}^{-1}(\mu)$ corresponding to a closed orbit $c_\mu(t) \in \mathcal{P} _\mu $ with period $T$ in the three-vortex problem.  Then its dynamic phase is
	\begin{align}
	  \theta _d (T) &= \int _0 ^T \alpha \left( X _h( c(t) ) \right) \, d t 
	  =  -\frac{ \Gamma _1 \Gamma _2 + \Gamma _2 \Gamma _3 + \Gamma _3 \Gamma _1 }{4 \pi \mu } \, T \,.  \label{eq:dynamicPhase}
	\end{align}
	That is to say, the dynamic phase corresponding to a closed orbit in $ \mathcal{P} _\mu $ is proportional to its period, the constant of proportionality being equal to the virial~(\ref{eq:virial}) divided by $ 2 \pi $ times the second moment of circulation~(\ref{eq:I2etc}).
\end{theorem}

\subsection{Total phase}

From~\eqref{geometricPhaseSimpleArea} and~\eqref{eq:dynamicPhase}, the total phase, which is the group composition of the geometric phase and the dynamic phase, can be expressed as 
\begin{align*}
\theta_{{\rm tot}}&= \theta_g + \theta_d, \quad {\rm (mod \: 2 \pi)} \\
&= \frac{ 1 }{ 2 }  \iint _{ U _1 }  \omega _1 + \frac{ V_0}{4 \pi I _2} \, T, \quad {\rm (mod \: 2 \pi)}
\end{align*}
where $ U _1 $ is the radial projection of $U$ on the normalized reduced space $ \mathcal{P} _1 $ and $ \omega _1 $ is the 2-form given in definition \ref{def:omegamu} with $ \mu = 1 $.

\subsection{Example}

Let us illustrate the theory developed so far by looking at the reconstruction phases associated to the case $ W _0 > 0 $ shown in figure \ref{fig:reducedspaces}.  Here $ \Gamma _1 = 7.615 $, $ \Gamma _2 = -3.46 $, $ \Gamma _3 = -3.155 $ and $ \mu = 1 $.

\begin{figure}
\begin{center}
\includegraphics[scale=0.6]{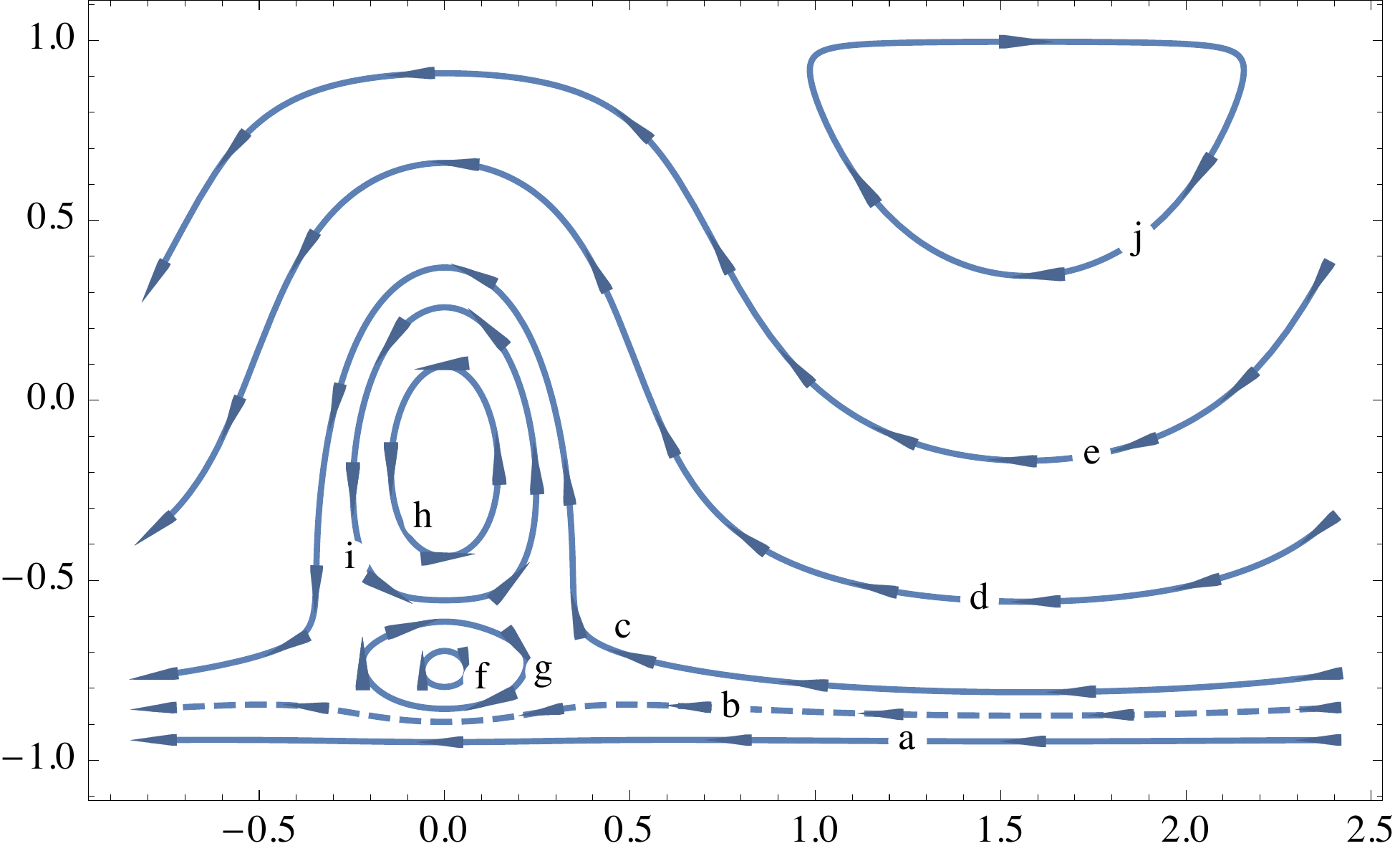}
\end{center}
\caption{\label{fig:phaseportrait}Phase portrait in cylindrical coordinates corresponding to the case $ W _0 > 0 $ shown in figure \ref{fig:reducedspaces}.  Horizontal axis is the interval $ - \pi / 4 \le \varphi _1 \le 3 \pi / 4 $.  Flow lines show the orientation induced by the reduced Hamiltonian vector field.}
\end{figure}

Figure \ref{fig:phaseportrait} shows a few flow lines of the phase portrait in cylindrical coordinates given by \eqref{cyl_coords_sphere}.  Note that all the flow lines are closed loops, since the left and right vertical lines at $ \varphi _1 = - \pi /4 $ and $ \varphi _1 = 3 \pi /4 $ actually coincide.

\begin{table}
\newcommand{\titlesi}{\text{orbit} & \text{energy} & \text{geom. phase} & \text{period} & \text{dyn. phase} & \text{total phase}\\\hline}
\begin{displaymath}
\begin{array}{c|ccccc}
\titlesi
 \text{(a)} & -11.9764 & 6.1127 & 0.0828
   & 0.26 & 0.0895 \\
 \text{(b)} & -10.1509 & 5.8607 & 0.2201
   & 0.6912 & 0.2687 \\
 \text{(c)} & -9.2487 & 4.8855 & 0.675 &
   2.1195 & 0.7218 \\
 \text{(d)} & -7.45 & 3.5096 & 0.9527 &
   2.9913 & 0.2177 \\
 \text{(e)} & -6.1434 & 1.9882 & 1.4105 &
   4.4289 & 0.1339 \\
 \text{(f)} & -7.45 & 0.0094 & 0.0107 &
   0.0337 & 0.0431 \\
 \text{(g)} & -9.2487 & 0.0828 & 0.1346 &
   0.4227 & 0.5054 \\
 \text{(h)} & -11.9764 & -0.1214 & 0.0646
   & 0.2027 & 0.0813 \\
 \text{(i)} & -10.1509 & -0.328 & 0.1875
   & 0.5888 & 0.2608 \\
 \text{(j)} & -5.2727 & 0.5784 & 1.8485 &
   5.8041 & 0.0993 \\
\end{array}
\end{displaymath}
\caption{\label{tab:reconstructionphases}Reconstruction phases for each of the periodic orbits shown in figure \ref{fig:phaseportrait}.  (Different periodic orbits may have the same energy.)}
\end{table}

Table \ref{tab:reconstructionphases} shows the reconstruction phases for each of the periodic orbits displayed in figure \ref{fig:phaseportrait}.  The geometric phase $ \theta _g $ was computed using formula \eqref{geometricPhaseDefined}.  The dynamic phase $ \theta _d $ is given by formula \eqref{eq:dynamicPhase}, where the period was computed numerically.  The total phase $ \theta _{\text{tot}} $ is $ \theta _g + \theta _d \;(\operatorname{mod} 2 \pi) $.  (All angles are measured in radians.)

\begin{figure}
\begin{center}
\includegraphics[scale=0.45]{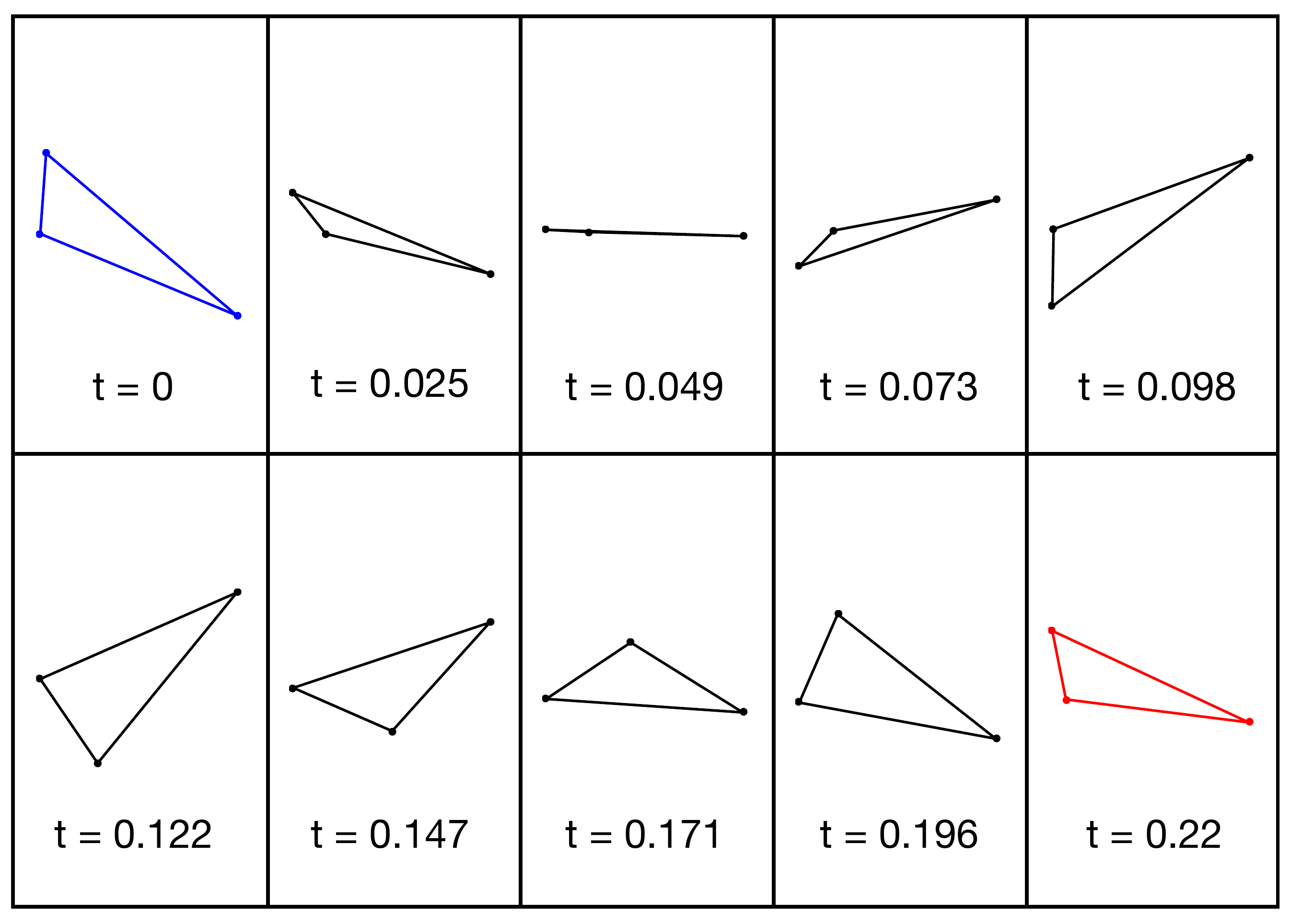}
\end{center}
\caption{\label{fig:vortexdyn}Evolution of the vortex triangle corresponding to the reduced periodic orbit labeled (b) in figure \ref{fig:phaseportrait}.  (Triangles have been rescaled to fit uniformly and actual position relative to the origin is not shown.)}
\end{figure}

For the periodic orbit (b), highlighted as a dashed curve in figure \ref{fig:phaseportrait}, the corresponding evolution of the vortex triangle is shown in figure \ref{fig:vortexdyn}.  Its initial and final configurations are displayed together in figure \ref{fig:vortexphases}, where $ \theta _g $, $ \theta _d $ and $ \theta _{\text{tot}} $ are represented.  Also shown in figure \ref{fig:vortexphases} is the center of circulation $ Z _0 $, placed at the origin.

\begin{figure}
\begin{center}
\includegraphics[scale=0.425]{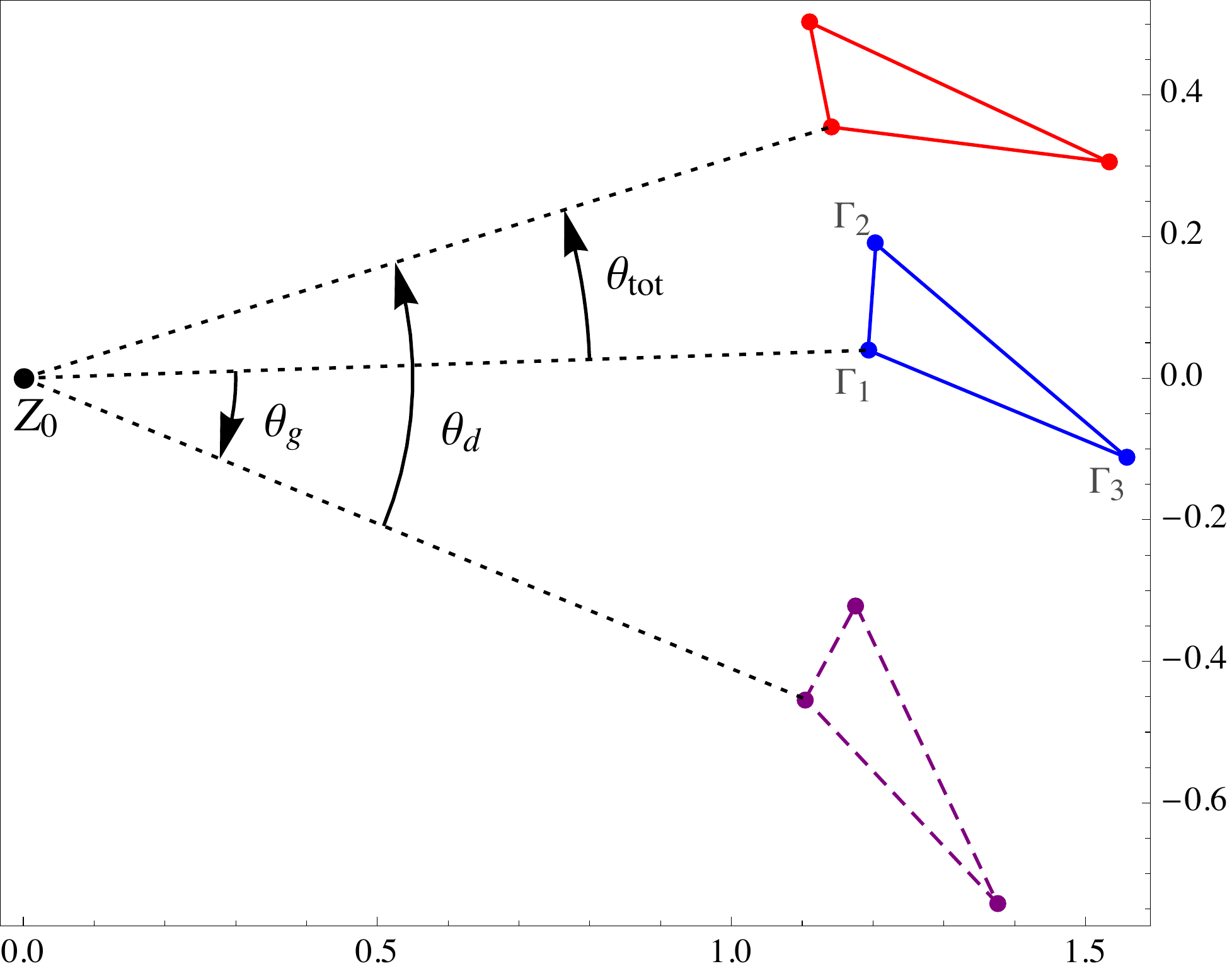}
\end{center}
\caption{\label{fig:vortexphases}Initial and final configurations for the time series shown in figure \ref{fig:vortexdyn}.  Also shown is the triangle obtained by rotating the initial configuration by the geometric phase $ \theta _g = 5.86072 = -0.42247 \; (\operatorname{mod} 2 \pi)$.  The dynamic and total phases are $ \theta _d = 0.69119 $ and $ \theta _{\text{tot}} = 0.26872 $.}
\end{figure}

\newpage

\subsection{The case of identical vortices}
Before ending this section, the case of identical vortices, $\Gamma_1=\Gamma_2=\Gamma_3=\gamma$, is discussed since it allows comparisons with two previous works in the literature on point vortices.

\paragraph{Comparison with the work of Aref and Pomphrey.}
The dynamic phase formula~(\ref{eq:dynamicPhase}) becomes more useful due an interesting exact solution derived by Aref and Pomphrey \cite{ArefPomphrey1982}. Introducing the variable $I:=(I_1/I_2)^2$, they show that the $(I_1,\varphi_1)$ system, being a two-dimensional Hamiltonian system, can be reduced to a quadrature which has an explicit solution in terms of Jacobi elliptic functions. This leads to an exact solution for the time period of the rotation in terms of elliptic integrals. Some details of their derivation are recalled in Appendix \ref{apdx:ArefPomphreyTimePeriod}, but the main result is as follows.

The time period of the $I$ variable is expressed by one of the following relations, 
\begin{align*}
T_a&= \frac{16 \pi \lambda^2 I_2 K(m_a)}{3 \gamma^2 \omega_a}, \quad  0 < \lambda < \frac{1}{\sqrt{2}}\\
 T_b&= \frac{16 \pi \lambda^2 I_2 K(m_b)}{3 \gamma^2 \omega_b}, \quad \frac{1}{\sqrt{2}} < \lambda < 1, 
\end{align*}
where the parameter $\lambda$ is defined as 
\[\lambda^2:=e^{-4 \pi H/\gamma^2}/(4 I_2^3)\]
and $K(m)$ is the complete elliptic integral of the first kind with $m$ the modulus of the elliptic function. The definitions of $m_a,m_b,\omega_a,\omega_b$ are given in Appendix \ref{apdx:ArefPomphreyTimePeriod}. The time period of the closed orbit of the $(I_1,\varphi_1)$ is some integral multiple $n$ of $T_a$ or $T_b$. Since all quantities in the $T_a,T_b$ relations are obtained from the initial configuration of the vortices, it follows that, {\it provided $n$ can be determined,  the dynamic phase~(\ref{eq:dynamicPhase}) can be predicted from the initial configuration of the identical vortices}.

\paragraph{Comparison with the work of Kuznetsov and Zaslavsky.}

In the interesting work of Kuznetsov and Zaslavsky \cite{KuZa98}, the chaotic advection of a tracer particle in the flow field of three identical point vortices is studied. In the appendix of that paper,  a formula for the rotation, i.e. the total phase, of the vortex  triangle is also derived. However, the formula is an integral formula and their approach is not in the framework of reconstruction phases and Hamiltonian systems with symmetry. Consequently, the insight that the rotation angle can be split into a dynamic phase and geometric phase is not obtained.

\section{Reconstruction phases with four identical vortices}
\label{sec:phasesFourVortices}

In the four vortex problem, the investigation of periodic orbits on the symmetry reduced spaces for arbitrary vortex strengths has not been as well studied. Aref and Pomphrey \cite{ArefPomphrey1982} show that such periodic orbits exist for the case of identical vortices exhibiting a discrete symmetry corresponding to rectangular configurations. These trajectories can be solved exactly, as mentioned in Novikov's earlier paper \cite{No75}. The discrete symmetry is preserved by the motion and this allows the reduction of the number of degrees of freedom.  In like manner, Bolsinov, Borisov and Mamaev \cite{BoBoMa1999} find some more periodic orbits for the four vortex problem with discrete symmetries.  Also, in the paper by Borisov, Mamaev and Killin \cite{BoMaKi2004}, several more relative (and absolute) periodic orbits for the case of identical vortices have been investigated.  Total rotation angles are computed in this paper but, as in \cite{BoBoMa1999}, the splitting into geometric and dynamic phases is not addressed.  
For the periodic orbits corresponding to identical vortices mentioned in the last two references, the reconstruction phases can in principle be computed using the formulas obtained in this section.

In what follows it is shown that the methods of the previous sections can also be applied to the four vortex problem with identical vortices.  Expressions for geometric and dynamic phases are obtained analogously to those in the three vortex problem.

\begin{remark}
	It is likely that the four vortex problem with non-identical vortices can be dealt with using Jacobi-Bertrand-Haretu coordinates, akin to what we did in section \ref{sec:phasesThreeVortices} for three vortices.  However, the aforementioned paucity of investigations of periodic orbits in this case make the problem less interesting.  For this reason, as well as for simplicity, we only consider identical vortices in this section.
\end{remark}

\subsection{Coordinate transformations}

Let us consider a system of four vortices, each with vortex strength $\gamma$, with position coordinates $ (z _1, z _2, z _3, z _4) \in \mathbb{C} ^4 $.  Introducing the matrix $\mathbb{H} _0$, write the Hermitian metric given in \eqref{eq:hermitianMetric} as
\[
	\left\langle \mathbf{z}, \boldsymbol{\zeta} \right\rangle = \mathbf{z} ^t \, \mathbb{H} _0 \, \boldsymbol{\zeta} \,, 
\] 
with $ \mathbb{H} _0 = \gamma \, \mathbb{I} _4 $.  Here $ \mathbb{I} _k $ is the $ k \times k $ identity matrix, $ \mathbf{z} = (z _1, z _2, z _3, z _4) $ and $ \boldsymbol{\zeta} = (\zeta _1, \zeta _2, \zeta _3, \zeta _4) $.

The coordinate transformations (similar to Aref and Pomphrey) for four identical vortices proceed as follows. The first canonical transformation 
$$T_1:(z_1,z_2,z_3,z_4) \mapsto \big(q _n + \mi \, p _n \big) _{ n = 0 } ^3 $$ 
is a discrete Fourier transform (DFT) given by
\begin{align*}
q_n + \mi \, p_n&=\frac{1}{\sqrt{N}} \sum_{\alpha= 1}^N z_{\alpha} \exp \left[\frac{\mi \, 2 \pi n ( \alpha - 1)}{N} \right], \quad n=0,1,2,3
\end{align*} 
with $N$=4.
By Parseval's theorem, the  DFT preserves the Hermitian form:  $ \left\langle T _1 \mathbf{z}, T _1 \boldsymbol{\zeta} \right\rangle = \left\langle \mathbf{z}, \boldsymbol{\zeta} \right\rangle $.  In other words, $$ \mathbb{H} _1 = \mathbb{H} _0 \,, $$ where $ \mathbb{H} _1 $ denotes the matrix of the Hermitian metric in the coordinates $ \{ q _n + \mi \, p_n \} $.  Note also that $ q _0 + \mi \, p _0 = (z _1 + z _2 + z _3 + z _4)/2 = 2 Z _0 $, where $ Z _0 $ is the center of circulation.

The second transformation $ T_2 : ( q _n + \mi \, p _n ) \mapsto ( j _n, \theta _n ) $ 
is given by 
\[\sqrt{2 j_n} \exp{i \theta_n}=q_n + \mi \, p_n, \quad n=0,1,2,3 \,. \] 
Note that, while $ T _2 : \mathbb{C} ^4 \longrightarrow \mathbb{R} ^8 $, $ T _2 ^{-1} $ can be extended to a map between two complex vector spaces
\[
	S _2 : \mathbb{C} ^8 \longrightarrow \mathbb{C} ^4 \,, 
\] 
so that $ T _2 ^{-1} = \left. S _2 \right| _{ \mathbb{R} ^8 } $.
The pull-back of the Hermitian metric in $ \mathbb{C} ^4 $ to $ \mathbb{C} ^8 $ is given by the matrix $ \mathbb{H} _2 := (D S _2) ^T \cdot \mathbb{H} _1 \cdot D S _2 $.  Towards obtaining its components, let $ \mathbf{r} = (r _0, r _1, r _2, r _3) $ with $ r _k = q_k + \mi \, p _k $.  We first compute 
\[
	\frac{ \partial r _k }{ \partial j _k } = \frac1{ \sqrt{ 2 j_k } } \me ^{ \mi \, \theta _k } \,, \quad \frac{ \partial r _k }{ \partial \theta _k } = \sqrt{ 2 j _k } \mi \, \me^{ \mi \, \theta _k } \,. 
\] 
Thus $ \left\langle \partial / \partial j _k, \partial / \partial j _k \right\rangle := \left\langle \partial \, \mathbf{r} / \partial j _k, \partial \, \mathbf{r} / \partial j _k \right\rangle = \gamma / (2 j _k) $.  Proceeding in a similar fashion it is easy to compute:
\[
	\left\langle \frac{ \partial }{ \partial j _k }, \frac{ \partial }{ \partial j _l } \right\rangle = \frac{ \delta'_{ kl } }{ 2 j _k } \,, \quad
	\left\langle \frac{ \partial }{ \partial j _k }, \frac{ \partial }{ \partial \theta _l } \right\rangle = -\mi \, \delta' _{ kl } \,, \quad
	\left\langle \frac{ \partial }{ \partial \theta _k }, \frac{ \partial }{ \partial \theta _l } \right\rangle =
	\delta' _{ kl } \, 2 j _k \,,
\] 
where $ \delta' _{ kl } = \gamma $ if $ k = l $ and zero otherwise.  Hence, 
the matrix of the Hermitian metric in coordinates $ (j _k, \theta _k) $, is given by 
\[
	\mathbb{H} _2 = \gamma \begin{pmatrix} \mathrm{J}^{-1} & - \mi \, \mathbb{I} _4 \\ \mi \, \mathbb{I} _4 & \mathrm{J} \end{pmatrix} \,, 
\quad \text{where} \quad 
	\mathrm{J} = 2 \begin{pmatrix} j _0 &&&\\  & j _1 &&\\ && j _2 &\\ &&& j _3 \end{pmatrix} \,. 
\] 

The third transformation $T_3:(j_n,\theta_n) \mapsto (I_n,\varphi_n)$ is achieved by means of the type II generating function
\[
G_2=-\varphi _0 j_0-j_1 \left(-2 \varphi _1-\varphi _2+\varphi _3\right)-j_2 \left(\varphi _1-\varphi _2+\varphi _3\right)-j_3 \left(\varphi _1+2 \varphi _2+\varphi _3\right)
\] 
which, with the equations $ \theta _k = \partial G _2 / \partial j _k $, $ I _k = \partial G _2 / \partial \varphi _k $, gives $ (I _0, \varphi _0) = (j _0, \theta _0) $ and
\begin{equation} \label{eq:4VMatrixEquationT3}
	\begin{pmatrix} \mathbf{I} \\ \boldsymbol{\varphi} \end{pmatrix} = \begin{pmatrix} A^{-1} &\\ &B^{-1} \end{pmatrix} \begin{pmatrix} \mathbf{j} \\[0.75ex] \boldsymbol{\theta} \end{pmatrix} \,, 
\end{equation}
where $ (\mathbf{j}, \boldsymbol{\theta}) = (j _1, j _2, j _3; \theta _1, \theta_2, \theta _3) $ and $ (\mathbf{I}, \boldsymbol{\varphi}) = (I _1, I _2, I _3; \varphi _1, \varphi _2, \varphi _3) $, and 
\[
	A = \frac13 \begin{pmatrix} -1 & 0 & 1 \\ 1 & -1 & 1 \\ 0 & 1 & 1 \end{pmatrix} \,, \quad
	B = \begin{pmatrix} -2 & -1 & 1 \\ 1 & -1 & 1 \\ 1 & 2 & 1 \end{pmatrix} \,. 
\] 
The generating function $ G _2 $ has been chosen so that $ A ^{t} B = \mathbb{I} _3 $, which is easily checked.

Note that $ I _0 = 2 | Z _0 | ^2 $, so that keeping the center of circulation at the origin amounts to $ I _0 = 0 $.  The corresponding indeterminate angle $\varphi_0$ is ignorable. Assuming this is the case, we only need to keep coordinates $ (\mathbf{I}, \boldsymbol{\varphi}) $ to describe the system.
The Hermitian metric in these coordinates is
\[
	\mathbb{H} _3 = 
	\gamma \begin{pmatrix} A ^t \\ & B ^t \end{pmatrix}
	\begin{pmatrix} \tilde{\mathrm{J}} ^{-1} & -\mi \, \mathbb{I} _3 \\ \mi\, \mathbb{I} _3 & \tilde{\mathrm{J}} \end{pmatrix}
	\begin{pmatrix} A \\ &B \end{pmatrix} =
	\gamma \begin{pmatrix} A ^t \tilde{\mathrm{J}} ^{-1} A & -\mi \, \mathbb{I} _3 \\ -\mi \, \mathbb{I} _3 & B ^t \tilde{\mathrm{J}} B \end{pmatrix} \,, 
\] 
where $ \tilde{\mathrm{J}} $ denotes the diagonal matrix $\mathrm{J}$ minus the $j_0$ entry, with the entries $ j _1, j _2, j _3 $ viewed as linear functions of $ I _1, I _2, I _3 $ according to \eqref{eq:4VMatrixEquationT3}.  For our purposes we only need to explicitly compute
\[
	B ^t \tilde{\mathrm{J}} B = 2 \begin{pmatrix} 
		-I _1 + I _3 & -I _1 + I _2 + I _3 & I _1 \\
		-I _1 + I _2 + I _3 & I _2 + 2 I _3 & I _2 \\
		I _1 & I _2 & I _3 
	\end{pmatrix} \,. 
\] 
	
The metric and symplectic form in coordinates $ (\mathbf{I}, \boldsymbol{\varphi}) $ can now be obtained from relations \eqref{eq:metricAndSymplecticForm}.  In this manner, from the expression for $ \mathbb{H} _3 $ given above we get 
\[
	\mathbb{M} _3 = \gamma \begin{pmatrix}
		A ^t \tilde{\mathrm{J}} A \\ & B ^t \tilde{\mathrm{J}} B
	\end{pmatrix}
	\quad \text{and} \quad
	\mathbb{J} _3 = \gamma \mathbb{J} 
\] 
where $ \mathbb{M}_3 $ and $ \mathbb{J} _3 $ are, respectively, the matrices of the metric $ \left\langle \! \left\langle \,, \right\rangle \! \right\rangle $ and symplectic form $\Omega$ in coordinates $ (\mathbf{I}, \boldsymbol{\varphi}) $.  (Here $ \mathbb{J} $ denotes the canonical symplectic matrix.)

From \eqref{eq:4VMatrixEquationT3} we compute that 
\[
	\varphi _1 = \frac{\theta _2 - \theta _1}{3} \,, \quad
	\varphi _2 = \frac{\theta _3 - \theta _2}{3} \,, \quad \text{and} \quad 
	\varphi _3 = \frac{\theta _1 + \theta_2 + \theta _3}{3} \,. 
\]
Thus, the rotation group $ SO(2) $ acts additively on the coordinate $ \varphi _3 $ while $ \varphi _1 $ and $ \varphi _2 $ remain invariant.  Since the hamiltonian is $ SO(2) $-invariant, $ \varphi _3 $ is a cyclic variable and thus $ I _3 $ is a conserved quantity.  Indeed, 
\begin{align*}
I_3&=j_1+j_2+j_3\\
&=\frac{1}{2}\left(\sum_{k=1}^3 q_k^2+\sum_{k=1}^3 p_k^2 \right) \\
&=\frac{1}{8}\left(4 \sum_{k=1}^4x_k^2+4 \sum_{k=1}^4y_k^2 - \left(\sum_{k=1}^4 x_k\right)^2 -\left(\sum_{k=1}^4 y_k\right)^2 \right) \\
&= \frac12 \sum _{ k = 1 } ^4 | z _k | ^2 = \Theta _0 / \gamma 
\end{align*}
where $ \Theta _0 $ is the angular impulse.  (Here we have used that the center of circulation is at the origin.)  
Using the expression for the symplectic matrix $ \mathbb{J}_3 $ given above, it is easy to compute that the momentum map $\mathbf{J}$ yielding $ X _{\mathbf{J}} = \partial / \partial \varphi_3 $ is minus the angular impulse; that is to say,
\[
	\mathbf{J} = - \gamma \, I _3 \,. 
\] 

Therefore the momentum level set $ \mathbf{J} ^{-1} (\mu) $ is parametrized by $(I_1 $, $I_2$, $\varphi_1$, $\varphi_2$, $\varphi_3)$.   The reduced space $ \mathcal{P} _\mu = J ^{-1} (\mu) / SO(2) $ is a four-dimensional manifold.  Coordinates $ ( I _1, I _2; \varphi _1, \varphi _2) \in \Delta \times \mathbb{T} ^2 $ define a coordinate chart for $ \mathcal{P} _\mu $, where $\Delta$ is the open triangle in $ \mathbb{R}^2 $ bounded by the lines  $ I _1 = - \mu / \gamma $, $ I _2 = - \mu / \gamma $ and $ I _2 - I _1 = \mu / \gamma $, and $ \mathbb{T}^2 = \mathbb{R} ^2 / (2 \pi \mathbb{Z} / 3) ^2 $.  We will refer to those points in $ \mathcal{P}_\mu $ not covered by this chart as its singular set $ \mathcal{S}_\mu $.

\vspace{1ex}
\begin{remark}
As a point in $ \mathcal{P}_\mu $ approaches $ \mathcal{S} _\mu $, $ (I _1, I _2) $ approaches $ \delta \Delta $.  If $ (I _1, I _2) \in \delta \Delta $, $ \{(I _1, I _2)\} \times \mathbb{T} ^2 $ gets mapped to a circle or a point, the latter being the case precisely when $ (I _1, I _2) $ is at a vertex of $\Delta$.
\end{remark}

\subsection{The connection}

The vertical space $ V _p $ at any $ p \in \mathbf{J} ^{-1} (\mu) $ is the linear space spanned by $ \partial / \partial \varphi _3 $.  The horizontal space $ H _p $ is metric orthogonal to the vertical space.  Let
\[
	w _p = a _1 \frac{ \partial }{ \partial I _1 } + a _2 \frac{\partial }{ \partial I _2 } + b _1 \frac{ \partial }{ \partial \varphi _1 } + b _2 \frac{ \partial }{ \partial \varphi _2 } + b _3 \frac{ \partial }{ \partial \varphi _3 } \,. 
\] 
Then, from the expression for $ \mathbb{M} _3 $ given above, the condition $ \left\langle \! \left\langle w _p, \partial / \partial \varphi _3 \right\rangle \! \right\rangle = 0 $ becomes
\[
	2 \gamma (b _1 I _1 + b _2 I _2 + b _3 I _3 ) = 0 \,. 
\] 
Thus,
\[
	H _p = \operatorname{span} \left( \frac{ \partial }{ \partial I _1 } \,, \; \frac{ \partial }{ \partial I _2 } \,, \; I _3 \frac{ \partial }{ \partial \varphi _1 } - I _1 \frac{ \partial }{ \partial \varphi_3 } \,, \; I _3 \frac{ \partial }{ \partial \varphi _2 } - I _2 \frac{ \partial }{ \partial \varphi_3 } \right) \,. 
\] 
The connection one-form $\alpha:T\mathbf{J}^{-1}(\mu) \rightarrow \mathbb{R}$, defined by the conditions $ \alpha (\partial / \partial \varphi _3 ) = 1 $ and $ w _p \in H _p \Rightarrow \alpha (w _p) = 0 $, is thus given by
\begin{equation} \label{eq:4VConnectionOneForm}
	\alpha = \frac{ I _1 }{ I _3 } d \varphi_1 + \frac{ I _2 }{ I _3 } d \varphi_2 + d \varphi_3 
\end{equation} 
or, setting $ I _3 = -\mu / \gamma $,
\[
	\alpha = - \frac{ \gamma }{ \mu } \left(  I _1 d \varphi _1 + I _2 d \varphi _2 \right) + d \varphi _3 \,. 
\]

\subsection{Reconstruction phases}

We now compute the geometric and dynamic phases.

\paragraph{The geometric phase.}  The connection 1-form can be written as 
\(
	\alpha = - \gamma \big( (I _1 - 1) d \varphi _1 + (I _2 + 1) d \varphi _2 + d \theta_2 \big) / \mu 
\).
Choose a section $ \sigma : \mathcal{P} _\mu \setminus \mathcal{S} _\mu \longrightarrow J ^{-1} (\mu) $ by setting $ \theta _2 = \text{const} $.  Then the pullback of the connection 1-form is
\[
	\sigma ^\ast \alpha = - \frac{ \gamma }{ \mu } \big[ (I _1 - 1) d \varphi _1 + (I _2 + 1) d \varphi _2 \big] 
\]
and the curvature of the connection is the 2-form
\[
	D \alpha := d \sigma ^\ast \alpha = - \frac{ \gamma }{ \mu } \big( d I _1 \wedge d \varphi _1  +  d I _2 \wedge d \varphi_2 \big) \,. 
\] 
(The same curvature is obtained, of course, if we use a section given by $ \theta _1 = \text{const.} $ or $ \theta _3 = \text{const} $.)  Therefore, 
by definition, the geometric phase around a closed orbit on the space $\mathcal{P}_\mu$ with coordinates $(I_1,I_2,\varphi_1,\varphi_2)$ is given by 
\begin{align*}
\theta_g:=\operatorname{exp} \left(-\oint_{c_\mu} \sigma^{\ast} \alpha \right) \equiv  \frac{\gamma}{\mu} \oint_{c _\mu} \big[ (I _1 - 1) d \varphi _1 + (I _2 + 1) d \varphi _2 \big]
\end{align*}
and, provided $c_\mu$ encloses a 2-dimensional subdomain $U \subset \mathcal{P}_\mu$, can also be expressed as 
\begin{align*}
\theta_g= \frac{\gamma}{\mu} \int_U \big(  dI_1 \wedge  d \varphi_1 + dI_2  \wedge d \varphi_2 \big) \,. 
\end{align*}

\paragraph{The dynamic phase.}  Let $ \delta _{ \alpha \beta } = z _\alpha - z _\beta $.  Going back to the discrete Fourier transform $ T _1 $, it is easy to see that, with the notation $ r _k = q _k + \mi \, p _k $,
\[
	\begin{pmatrix} r _1 \\ r _2 \\ r _3 \\ 0 \end{pmatrix} = \frac12 \begin{pmatrix}
		0 & 1 & \mi & 0 \\
		1 & 0 & 0 & 1 \\
		0 & 1 & -\mi & 0 \\
		1 & -1 & 1 & -1
	\end{pmatrix} \begin{pmatrix} \delta _{ 12 } \\ \delta _{ 13 } \\ \delta _{ 24 } \\ \delta _{ 34 } \end{pmatrix} \,, 
\] 
and that this is an invertible linear relation.  Moreover, $ \delta _{ 14 } = \delta _{ 12 } + \delta_{ 24 } $ and $ \delta _{ 23 } = \delta _{ 13 } - \delta _{ 12 } $.  It follows that the hamiltonian can be expressed as
\begin{align*}
	h &= - \frac1{4 \pi} \sum _{ \alpha < \beta } \Gamma _\alpha \Gamma _\beta \ln | \delta _{ \alpha \beta }  | ^2 \\
	&= - \frac{ \gamma ^2 }{ 4 \pi } \sum _{ \alpha < \beta } \ln \left| \sum _{ k = 1 } ^3  c _{ \alpha \beta } ^k  \, r _k \right| ^2 \,, \quad c _{ \alpha \beta } ^k \in \mathbb{C} \,. 
\end{align*}
Since $ r _k = \sqrt{2 j _k} \, \me ^{ \mi \, \theta _k } $ and $ j _k $ depends linearly on $ I _1, I _2, I _3 $, we have
\[
	h = - \frac{ \gamma ^2 }{ 4 \pi } \sum _{ \alpha < \beta } \ln \left( b _{ \alpha \beta } (I _1, I _2, I _3; \varphi _1, \varphi_2, \varphi_3 ) \right) 
\] 
where the $ b _{ \alpha \beta } $ are homogeneous functions of degree 1 in the variables $ I _1, I _2, I _3 $.

The dynamic phase is obtained as before by first computing 
\begin{align*}
\xi(t)&=\alpha \left(X_H(c(t))\right) 
\end{align*}
where $c(t)$ is the system trajectory on $\mathbf{J}^{-1}(\mu)$.  Using the expression for the connection 1-form in \eqref{eq:4VConnectionOneForm} and recalling that in $ (\mathbf{I}, \boldsymbol{\varphi}) $ coordinates the symplectic matrix is
\(
	\mathbb{J} _3 = \gamma \begin{pmatrix} 0 & \mathbb{I}  \\ - \mathbb{I}  & 0 \end{pmatrix}
\)
we get 
\begin{align*}
	\xi(t)&=\frac{ I_1}{ I _3 } \frac{d \varphi_1}{dt} + \frac{I_2}{ I _3 } \frac{d \varphi_2}{dt} +  \frac {d \varphi_3}{dt}, \\
	&=- \frac{ 1}{\gamma I_3} \left(I_1 \frac{\partial h}{\partial I_1} + I_2 \frac{\partial h}{\partial I_2} + I_3  \frac {\partial h}{\partial I_3} \right) \\
	&= \frac{ 1}{\mu} \, \mathbf{I} \cdot \nabla _I \, h(\mathbf{I}, \boldsymbol{\varphi}) \\
	&= -\frac{ \gamma ^2 }{ 4 \pi \mu } \sum _{ \alpha < \beta } \mathbf{I} \cdot \nabla _I  \ln \left( b _{ \alpha \beta } (\mathbf{I}; \boldsymbol{\varphi}) \right)
\end{align*} 
where $ \nabla _I $ denotes the gradient with respect to the $ I _1, I _2, I _3 $ variables.  Therefore, applying corollary \ref{cor:Euler},
\[
	\xi(t) = -\frac{ \gamma ^2 }{ 4 \pi \mu } \sum _{\alpha < \beta } 1 = -\frac{ 3 \gamma ^2 }{ 2 \pi \mu } \,. 
\] 
We thus arrive at

\begin{theorem} \label{thm:dynamicPhase4V}
	Let $ c (t) $ be the orbit of the system in $\mathbf{J}^{-1}(\mu)$ corresponding to a closed orbit $c_\mu(t) \in \mathcal{P}_\mu$ with period $T$ in the problem of four identical vortices with center of circulation at the origin.  Then its dynamic phase is
	\[
	  \theta _d (T) = \int _0 ^T \alpha \left( X _h(c (t) ) \right) \, d t 
	  =  -\frac{3  \gamma ^2 }{2 \pi \mu } \, T \,. 
	\] 
	That is to say, the dynamic phase corresponding to a closed orbit in $\mathcal{P}_\mu$ is proportional to its period, the constant of proportionality being equal to the virial divided by $ 2 \pi $ times the second moment of circulation.
\end{theorem}

\begin{remark}
	Note that, phrased in terms of the virial and the second moment of circulation, the formula for the dynamic phase for both the three- and four- vortex problems coincide (at least for equal vortex strengths).  Cf. theorems \ref{thm:dynamicPhase} and \ref{thm:dynamicPhase4V}.
\end{remark}

\section{Summary and future directions.} 
\label{sec:summary}

In this paper, pure reconstruction phases, which are drifts in symmetry group directions obtained at the end of symmetry reduced periodic orbits, are computed for the Hamiltonian systems of (a) three point vortices of arbitrary strengths in the plane
and (b) four point vortices of identical strengths in the plane, respectively. Such phases are distinguished from the adiabatic phases made popular by the works of Berry and Hannay in that they do not require a (infinitely) slow evolution of any system parameter  and are obtained `purely' by the method of reconstruction--the opposite of the method of symmetry reduction. As per the general theory, relative to the choice of a connection the total phase can be viewed as being composed of a geometric phase and a dynamic phase. The geometric phase in both models is shown to be expressible as integrals over the domains enclosed by the periodic orbits on the symmetry reduced spaces. In model (a), at least, the integral is shown to be related to the area enclosed by the closed loop.

The main contribution of this paper, however, is obtaining simple exact formulas that relate the dynamic phase in each model to the time period of the periodic orbit. Dynamic phases are typically expressed as time integrals involving non-trivial integrands.  Simple exact formulas for the dynamic phase are rare in the literature on phases. A notable example being Montgomery's formula for the dynamic phase in the motion of a free rigid body \cite{Mo91}. The homogeneity of the Hamiltonian function plays an important role in obtaining these simple formulas but it is not enough. The choice of the connection 1-form and the properties of the isotropy subgroup are also important.  More precisely, the simplicity of our formula for the dynamic phase seems to be a consequence of the fact that \emph{a}) the Hamiltonian is a $ \log $-homogeneous function of the distances between the particles (vortices), \emph{b}) the metric and symplectic forms are compatible (the configuration space is a K\"ahler manifold), and \emph{c}) the coadjoint-isotropy subgroup is Abelian.

As remarked at the beginning of section \ref{sec:phasesFourVortices}, the methods of this paper should allow to obtain formulas for the reconstruction phases in the case of non-identical four vortices.  These could be applied to the various periodic orbits found by exploiting discrete symmetries, as in the papers mentioned at the beginning of section \ref{sec:phasesFourVortices}.  It is noteworthy that discrete symmetries also play a prominent role in periodic orbits found in celestial mechanics; see e.g. \cite{FePo2008} and references therein.  It is thus intriguing to explore how various discrete symmetries get manifested in the obtention of reconstruction phases for such systems.

Since the pioneering work of Kirchhoff and Helmholtz, the $N$-point-vortex model in the plane has seen several developments and extensions. In particular, extensions to $\mathbb{R}^3$ where the vorticity resides on curves---vortex filaments/rings--have been widely studied. An interesting history of work in this area (till the year 1996) may be found in \cite{Ri1996}. Discrete vortex models have also found applications in diverse areas of physics and mechanics, extending well beyond classical fluids, for example, in superfluids and magnetohydrodynamics, and also to problems of quantization of classical vortices. A few representative examples of the large amount of work in these areas are \cite{Fe1966, Ne1990, KiTaPeMc1993, KiTaMcPe1994, GoMeSh1987}. More recently, models in higher dimensional spaces, $\mathbb{R}^4$ and higher, for both classical and superfluids have also been developed \cite{Li1998, Je2002, HaVi2003, Sh2012, Kh2012}. 

    A future endeavor would be to apply the techniques of this paper to these extensions. Hamiltonian models of such systems (in spaces $\mathbb{R}^3$ and higher) are in general infinite dimensional. Some special configurations, such as axisymmetric circular vortex rings, allow a finite dimensional representation and indeed reconstruction phases have been computed for such models \cite{ShMa2003}. Numerical simulations in the papers \cite{RyLe1997, JeSm2011}, which have different objectives from this work, strongly suggest that reconstruction phases may exist even in infinite dimensional models of vortex rings. Basic ideas used in this paper such as the metric, the symmetry groups and the connection 1-form can presumably be carried over to such models. However, the canonical transformations and the characterization of the symmetry reduced spaces are ideas which greatly benefit from the finite dimensionality of the point vortex model. Extending these to infinite dimensional models could prove to be challenging. A more promising direction is to apply the ideas of this paper to $N$-point-vortex models developed for 2-dimensional compact boundaryless surfaces; see, for example, \cite{BoKo2015} which also contains a good list of references.  These are again finite dimensional systems, with the underlying curvature of the surface and its compactness strongly influencing the dynamics of the point vortices. It would be interesting to study reconstruction phases in these models for cases where there exist symmetries and symmetric periodic orbits. 

\paragraph{Acknowledgments.} BNS is very grateful to Matthew Perlmutter for several discussions on the topic of this paper.  Both authors wish to thank the anonymous referees for their helpful comments.

%:APPENDICES
\appendix

\section{Appendix: Independence of the choice of a JBH chart.}
\label{apdx:JBH_charts}

In this appendix we establish that, in the context of the three-vortex problem, formula \eqref{geometricPhaseSimpleArea}  for the geometric phase in terms of the area enclosed by a loop $ c _\mu $ is indeed valid even when $ c _\mu $ is not homotopic to a point without passing through a singular point of the coordinate chart.  Hence the geometric phase is independent of the different possible choices of a Jacobi-Bertrand-Haretu (JBH) coordinate chart.

In section \ref{sec:jacobiCoordinates} we introduced the JBH coordinates by first considering the vector from vortex \emph{one} to vortex \emph{two}, and then the vector from their center of circulation to vortex \emph{three}.  This coordinate chart was denoted $ T _1 $, but for the purposes of the following discussion we will add the superscript ``$(3)$'', to remind us that in this chart  vortex three is the last one taken into account.  This choice induces the coordinate chart $ \psi ^{(3)}: P _0 \rightarrow \mathbb{R} ^2 \times \mathbb{T} ^2$ given by $ \psi ^{(3)} = T _3 \circ T _2 \circ T _1 ^{ (3) } $, with $ T _2 $ and $ T _3 $ the transformations introduced in sections \ref{sec:AACoords}--\ref{sec:rotationalSymmetry}.  Let us also add the superscript to the coordinates furnished by $ \psi ^{(3)} $, so that we write
\[
	\psi ^{ (3) } (\mathbf{z}) = \left( T _3 \circ T _2 \circ T _1 ^{ (3) } \right) (\mathbf{z}) = \left( I _1 ^{ (3) }, I _2 ^{ (3) }; \varphi _1 ^{ (3) }, \varphi _2 ^{ (3) } \right) \,,
\] 
with $ \mathbf{z} = (z _1, z _2, z _3) \in P _0 \subset \mathbb{C} ^3 $.

Likewise there are two other JBH-charts that can be considered, namely $ T _1 ^{ (k) } $, $ k = 1, 2 $, constructed by first taking the vector $ r ^{ (k) } $ from vortex $i$ to vortex $j$ and then the vector $ s ^{ (k) } $ from their centers of vorticity to vortex $k$.  The coordinate chart $ \psi ^{(k)} = T _3 \circ T _2 \circ T _1 ^{ (k) } $ is thus defined, yielding the coordinates $ \left( I _1 ^{ (k) }, I _2 ^{ (k) }; \varphi _1 ^{ (k) }, \varphi _2 ^{ (k) } \right) $, $ k = 1, 2 $.  (The transformations $ T _2 $ and $ T _3 $ remain unchanged.)

If $ W _0 > 0 $, each coordinate chart $ \psi ^{(k)} $ has two \emph{singular points} associated to it, corresponding to $ I ^{(k)}_1 = \pm \mu $, $ \varphi ^{(k)}_1 \in [0, \pi] $.  These correspond to i) binary collision of vortices $i$ and $j$, and ii) vortex $k$ positioned at the center of circulation of vortices $i$ and $j$.   On the sphere $ \mathcal{P} _\mu $ these correspond to two antipodal points.  Note that the singular points of coordinate charts $ \psi ^{ (1) }, \psi ^{ (2) }, \psi ^{ (3) } $ all lie on the same great circle.  If $ W _0 < 0 $ then each coordinate chart $ \psi ^{(k)} $ has only one singular point associated to it, corresponding to either binary collision or a collinear configuration.

Proposition \ref{prop:geometric_phase_JBH_charts} below effectively shows that formula \eqref{geometricPhaseSimpleArea} for the geometric phase does not depend on the choice of coordinate chart.  Before stating it let us discuss some preliminary notions and lemmas.

\subsection*{Orientation}

Let us recall some basic facts about orientable two-dimensional manifolds and establish some terminology.

Let $M$ be an oriented 2-dim manifold and $\omega$ a non-degenerate 2-form defined on it.  We say that $\omega$ is compatible with the orientation of $M$ if for every $ p \in M $ and every pair of vectors $ u, v \in T _p M $, $ \omega(u,v) > 0 $ iff $ (u,v) $ are positively oriented.
In this manner, given an orientable 2-dim manifold $M$, a non-degenerate 2-form $\omega$ on $M$ induces an orientation of $M$.

Let $R$ be a region in $M$ with smooth boundary $ \delta R $.  An orientation in $ \delta R $ is induced by the following requirement:  if $ \gamma(t) $ is a (local) parametrization of $ \delta R $ compatible with its orientation then, for every $ u \in T _{ \gamma(t) } M $, the pair $ (\dot{\gamma} (t), u) $ is positively oriented precisely when $u$ points inward $R$.  The orientation of $ \delta R $ thus defined is \emph{the orientation of $ \delta R $ induced by $R$} or \emph{the positive orientation of $ \delta R $ with respect to $R$}.

If $ \gamma = \delta R $ with a given orientation, we say that \emph{$R$ is compatible with $\gamma$} if $\gamma$ is positively oriented with respect to $R$.
Finally, if $\gamma$ is an oriented curve, $ - \gamma $ denotes the same curve with the opposite orientation.

\subsection*{Geometric phase formula}

We first consider the compact case $ W _0 > 0 $.
Let $ S ^2 _\rho $ denote the 2-dim sphere of radius $ |\rho| \neq 0 $ and $ S ^2 $ the unit sphere.

\begin{definition} \label{def:unitaryprojection}
Given $U$ a region on $ S ^2_\rho $, let 
\[
	\mathcal{A} (U) := \frac12 \iint _{ U _1 } \omega _1 
\] 
where $ U _1 $ is the concentric radial projection of $U$ on $ S ^2 $ and $ \omega _1 $ is the standard area element on the unit sphere.  In other words, \emph{$ \mathcal{A} (U) $ is the semi-area of the projection of $U$ on the unit sphere}.
\end{definition}

\begin{lemma} \label{lem:complementSphere}
Let $U$ be a region on $ S ^2 _\rho $.  Then $ \mathcal{A} (S ^2 _\rho \setminus U) = 2 \pi - \mathcal{A} (U) $.
\end{lemma}

\begin{proof}
Follows readily from $ \mathcal{A} (S ^2) = 2 \pi $.
\end{proof}

Recall that:  i) since we are assuming $ \mu \neq 0 $ and $ W _0 > 0 $, we have that $ \mathcal{P} _\mu $ coincides with $ S ^2 _{ \mu } $;  ii) we have chosen a global section $\sigma$ on $ \mathcal{P} _\mu \setminus \{N, S\} $, where $N$ and $S$ are the singular points of chart $ \psi ^{ (k) } $, with respect to which $ \sigma ^\ast \alpha = (1 - I _1 / \mu) \, d \varphi _1 $, where $\alpha$ is the connection 1-form defined in \eqref{conn1form3v}; and iii) we have chosen the orientation on $ \mathcal{P} _\mu $ induced by the 2-form $ \omega _\mu = 2 \mu \, d I _1 \wedge d \varphi _1 $ (see definition \ref{def:omegamu}).

\begin{lemma} \label{lem:stokesProj}
Let $R$ be a region in $ \mathcal{P} _\mu $ (not necessarily simply connected), such that its closure does not contain a singular point of chart $ \psi ^{ (k) }$.  Let $ \delta R $ be positively oriented with respect to $R$.  Then 
\[
	-\int _{ \delta R } \sigma ^\ast \alpha = \mathcal{A} (R)  \,. 
\] 
\end{lemma}

\begin{proof}  By Stokes' theorem,
\begin{align*}
	-\int _{ \delta R } \sigma ^\ast \alpha &= \int _{ \delta R } \left( -1 + \frac{ I _1 }{ \mu } \right) \, d \varphi_1 = \frac{ \sign(\mu) }{ \mu }  \iint _{R} d I _1 \wedge d \varphi _1 \\
	&= \frac1{2 \mu ^2}  \iint _R 2 | \mu | \, d I _1 \wedge d \varphi _1 = \frac1{2 \mu ^2} \iint _R \omega _{ | \mu | } \\
	&= \frac12  \iint _{ R _1 } \omega _1 = \mathcal{A} (R) \,, 
\end{align*}
where $ R _1 $ is the concentric radial projection of $R$ on the unit sphere.  In the second equality, the factor $ \sign(\mu) $ has taken care of the fact that $ \omega _\mu $ and $ d I _1 \wedge d\varphi _1 $ induce opposite orientations when $\mu < 0$, and in the second to last equality we have used the identity $ \omega _{ | \mu | } = \mu ^2 \omega _1 $.
\end{proof}

\begin{proposition} \label{prop:geometric_phase_JBH_charts}
Let $ c _\mu $ be an oriented simple closed curve in $ \mathcal{P} _\mu $.  Assume that $ c _\mu $ does not pass through a singular point of chart $ \psi ^{(k)} $.  Let $U$ be the region in $ \mathcal{P} _\mu $ bounded by $ c _\mu $ and compatible with its orientation.  Then
\[
	-\oint _{ \delta R } \sigma ^\ast \alpha = \mathcal{A} (U) \qquad (\operatorname{mod} 2 \pi) \,. 
\] 
\end{proposition}

\begin{proof}
	We consider two possible cases.
	
	\emph{Case 1:}  $ c _\mu $ is contractible to a point without going through a singular point of chart $ \psi ^{(k)} $.  Let $n$ be the number of singular points inside $U$.  Then $ n = 0 $ or $ n = 2 $.  Let
\[
	U ^\ast = \left\{ \begin{array}{c@{\quad\text{if}\quad}r} U & n = 0 \\ \mathcal{P} _\mu \setminus U & n = 2 \end{array} \right. \,. 
\] 
Hence $ U ^\ast $ contains no singular points and $ (-1) ^{ n/2 } c _\mu $ is positively oriented with respect to $ U ^\ast $.
By lemma \ref{lem:stokesProj},
\begin{displaymath}
	-\oint _{ c _\mu } \sigma ^\ast \alpha = - (-1) ^{ n/2 } \oint _{ (-1) ^{ n/2 } c _\mu } \sigma ^\ast \alpha 
	= (-1) ^{ n/2 } \mathcal{A} ( U ^\ast ) \,, 
\end{displaymath}
which equals $ \mathcal{A} (U) $ if $ n = 0 $ and, by lemma \ref{lem:complementSphere}, equals $ - 2 \pi + \mathcal{A} (U) $ if $ n = 2 $.  Thus the claim follows.

	\emph{Case 2:}  $ c _\mu $ is not contractible to a point without going through a singular point of chart $ \psi ^{ (k) }$.  Say, without loss of generality, that $U$ contains the ``north pole'' $ N = (I _1 = \mu, \varphi _1 \in [0, \pi]) $.  Let $ \kappa _\varepsilon $ be the circle $ \{\mu - \varepsilon\} \times [0, \pi] $, with $ \varepsilon > 0 $ small enough so that $ \kappa _\varepsilon \subset U $.  Let $ V _\varepsilon $ be the region bounded by $ \kappa _\varepsilon $ containing $N$.  Orient $ k _\varepsilon $ positively with respect to $ V _\varepsilon $.  Then, by lemma \ref{lem:stokesProj},
	\begin{align*}
	 - \oint _{ c _\mu } \sigma ^\ast \alpha &= -\int _{ c _\mu \cup (- \kappa _\varepsilon) } \sigma ^\ast \alpha
		\; - \;
		\oint _{ \kappa _\varepsilon } \sigma ^\ast \alpha  \\
	&= \mathcal{A} ( U \setminus V _\epsilon ) - \int _0^{ \pi } \left( 1 - \frac{ \mu - \varepsilon }{ \mu } \right) \, d\varphi _1 \,. 
	\end{align*}
Taking the limit as $ \epsilon \rightarrow 0 $ the claim follows.
\end{proof}

Let us now consider the non-compact case $ W _0 < 0 $.  Let $ c _\mu $ be a simple closed curve on $ P _\mu $, now identified with one sheet of the hyperboloid $ H ^2 _\mu $.  Let $U$ be the region bounded by $ c _\mu $.  Assume that $ c _\mu $ is positively oriented with respect to $U$.  Definition \ref{def:unitaryprojection} is trivially modified with $ U _1 $ now the radial projection on the normalized hyperboloid $ H ^2 _1 $.  The proof of lemma \ref{lem:stokesProj} and proposition \ref{prop:geometric_phase_JBH_charts} go through without changes, except that in case 1 of the proof of proposition \ref{prop:geometric_phase_JBH_charts} only $ n = 0 $ needs to be considered (and thus lemma \ref{lem:complementSphere} is not needed).  If $ c _\mu $ is negatively oriented with respect to $U$, a minus sign is added to the formula in proposition \ref{prop:geometric_phase_JBH_charts}.

\begin{remark}
	The case when $ c _\mu $ passes through a singular point of a coordinate chart need not be considered, since these are either equilibrium points or unreachable points of the vortex dynamics.
\end{remark}

\section{Appendix: The Aref-Pomphrey expressions for the time period in the case of three identical vortices.}
\label{apdx:ArefPomphreyTimePeriod}

   The ODE satisfied by the variable  $I:=(I_1/I_2)^2$ is 
\begin{align}
\left( \frac{dI}{d \tau} \right)^2&=-I \left(I^3+6I^2+I(9-24\lambda^2) +8 \lambda^2(2 \lambda^2- 1) \right), \label{eq:i2eqn}
\end{align}
where $\lambda^2:=e^{-4 \pi H/\gamma^2}/(4 I_2^3)$, and 
\[\tau:= \frac{3 \gamma^2}{8 \pi \lambda^2 I_2} t\]
 is time rescaled for the trajectory. Since $H$ and $I_2$ are conserved quantities, $\lambda^2$ is constant for a trajectory on the symplectic reduced space and is fixed by the initial conditions. Note that the $\Lambda^2$ and $\tau$ defined in \cite{ArefPomphrey1982} are not exactly the same due to the following slight (and unimportant) differences: the definitions of the $(I_1, I_2, \varphi_1, \varphi_2)$ coordinates and the assumption of non-unit vortex strength $\Gamma$.
 
Aref and Pomphrey \cite{ArefPomphrey1982} show that the solution of~(\ref{eq:i2eqn}), depending on the value of $\lambda$, is one of the following: 
\begin{align}
I_{(a)}(\tau;\lambda)&=\frac{\mathcal{I}_0-\mathcal{I}_1 \kappa_a \operatorname{sn}^2 (\omega_a \tau)}{1-\kappa_a \operatorname{sn}^2 (\omega_a \tau) }, \quad 0 < \lambda < \frac{1}{\sqrt{2}}, \label{eq:ia} \\
I_{(b)}(\tau;\lambda)&=\frac{\mathcal{I}_0-\mathcal{I}_1 \kappa_b \operatorname{sn}^2 (\omega_b \tau)}{1-\kappa_b \operatorname{sn}^2 (\omega_b \tau) }, \quad \frac{1}{\sqrt{2}} < \lambda < 1 , \label{eq:ib}
\end{align}
where $\operatorname{sn}$ denotes the Jacobi elliptic sn function, 
\begin{align*}
\kappa_a&=\frac{\mathcal{I}_0}{\mathcal{I}_1}, \quad \kappa_b=\frac{\mathcal{I}_2-\mathcal{I}_0}{\mathcal{I}_2-{\mathcal{I}_1}}, \\
\omega_a&=\frac{1}{2}\sqrt{\frac{\mathcal{I}_2-\mathcal{I}_0}{\mathcal{I}_1}}, \quad \omega_b=\frac{1}{2}\sqrt{\frac{\mathcal{I}_0}{\mathcal{I}_2-\mathcal{I}_1}}
\end{align*} 
and $\mathcal{I}_n(\lambda), n=0,1,2$ are the three roots of the cubic polynomial in~(\ref{eq:i2eqn}) for which Aref and Pomphrey give explicit expressions  \cite{ArefPomphrey1982}.

   From the theory of elliptic functions and integrals,~(\ref{eq:ia}) and~(\ref{eq:ib}) are periodic with period equal to twice the complete elliptic integral of the first kind $K(m)$, where $m$ is the modulus of the elliptic function. Transforming to time $t$, the time period of the $I$ variable is obtained as  
\[T_a= \frac{16 \pi \lambda^2 I_2 K(m_a)}{3 \gamma^2 \omega_a}, \quad T_b= \frac{16 \pi \lambda^2 I_2 K(m_b)}{3 \gamma^2 \omega_b}\]
with \[m_a^2=\frac{\mathcal{I}_0(\mathcal{I}_2-\mathcal{I}_1)}{\mathcal{I}_1(\mathcal{I}_2-\mathcal{I}_0)}, \quad m_b^2=\frac{\mathcal{I}_1(\mathcal{I}_2-\mathcal{I}_0)}{\mathcal{I}_0(\mathcal{I}_2-\mathcal{I}_1)}\]
In general, some integral multiple of these time periods will be the time period $T$ of the periodic orbit on the symplectic reduced space---the $(I_1,\phi_1)$ space.  Note that $I_1$ is proportional to the oriented area of the vortex triangle but $I:=(I_1/I_2)^2$ masks orientation information.

%:Bibliography

\end{document}